\documentclass[moor,copyedit,arxiv]{informs3}

\usepackage{dsfont}	
\usepackage{longtable}
\usepackage{mathrsfs}
\usepackage{mathtools}
\usepackage[shortlabels]{enumitem}

\usepackage{tikz}
\usetikzlibrary{calc,patterns}
\usetikzlibrary{arrows,shapes,positioning}
 
\newcommand{\argdot}{\,\cdot\,}

\newcommand{\fakeopt}[1]{\widetilde{\raisebox{0pt}[0.85\height]{$\widetilde{#1}$}}}

\DeclareMathOperator{\Bernoulli}{\mathsf{Bernoulli}}
\DeclareMathOperator{\Bg}{\mathsf{B}}
\DeclareMathOperator{\Binomial}{\mathsf{Binomial}}
\DeclareMathOperator{\BNE}{\mathsf{BNE}}
\DeclareMathOperator{\Eq}{\mathsf{Eq}}
\DeclareMathOperator{\ESC}{\mathsf{ESC}}
\DeclareMathOperator{\Expect}{\mathsf{E}}
\DeclareMathOperator{\MNE}{\mathsf{MNE}}
\DeclareMathOperator{\Multi}{\mathsf{Multinomial}}
\DeclareMathOperator{\Opt}{\mathsf{Opt}}
\DeclareMathOperator{\Pg}{\mathsf{P}}
\DeclareMathOperator{\PNE}{\mathsf{PNE}}
\DeclareMathOperator{\Prob}{\mathsf{P}}
\DeclareMathOperator{\PoA}{\mathsf{PoA}}
\DeclareMathOperator{\Poisson}{\mathsf{Poisson}}
\DeclareMathOperator{\PoS}{\mathsf{PoS}}
\DeclareMathOperator{\SC}{\mathsf{SC}}
\DeclareMathOperator{\tot}{tot}
\DeclareMathOperator{\TV}{\mathsf{TV}}
\DeclareMathOperator{\Var}{\mathsf{Var}}
\DeclareMathOperator{\WE}{\mathsf{WE}}
\DeclareMathOperator{\Wg}{\mathsf{W}}

\DeclarePairedDelimiter{\braces}{\{}{\}}
\DeclarePairedDelimiter{\bracks}{[}{]}
\DeclarePairedDelimiter{\parens}{(}{)}

\DeclarePairedDelimiter{\abs}{\lvert}{\rvert}
\DeclarePairedDelimiter{\norm}{\lVert}{\rVert}

\usepackage{natbib}
 \NatBibNumeric
 \def\bibfont{\small}%
 \bibpunct[, ]{[}{]}{,}{n}{}{,}%

\usepackage[colorlinks=true,breaklinks=true,bookmarks=true,urlcolor=blue,
     citecolor=blue,linkcolor=blue,bookmarksopen=false,draft=false]{hyperref}

\def\EMAIL#1{\href{mailto:#1}{#1}}
\def\URL#1{\href{#1}{#1}}         

\TheoremsNumberedThrough     

\EquationsNumberedThrough    

\begin{document}


 \RUNAUTHOR{Cominetti et al.} 

\RUNTITLE{Approximation and Convergence of Large Atomic Congestion Games}

\TITLE{Approximation and Convergence of Large Atomic Congestion Games}

\ARTICLEAUTHORS{%
\AUTHOR{Roberto Cominetti}
\AFF{Facultad de Ingenier\'ia y Ciencias, Universidad Adolfo Ib\'a\~nez, Santiago, Chile, \EMAIL{roberto.cominetti@uai.cl}, \URL{https://sites.google.com/site/cominettiroberto}}
\AUTHOR{Marco Scarsini}
\AFF{Dipartimento di Economia e Finanza, Luiss University, Roma, Italy, \EMAIL{marco.scarsini@luiss.it}, \URL{https://sites.google.com/view/marcoscarsini}}
\AUTHOR{Marc Schr\"oder}
\AFF{School of Business and Economics, Maastricht University, Maastricht, The Netherlands, \EMAIL{m.schroder@maastrichtuniversity.nl}, \URL{https://www.maastrichtuniversity.nl/p70024063}}
\AUTHOR{Nicol\'as Stier-Moses}
\AFF{Core Data Science, Meta, Menlo Park, USA, \EMAIL{nicostier@yahoo.com}, \URL{https://sites.google.com/site/nicostier}}
} 

\ABSTRACT{%
We consider the question of whether, and in what sense, Wardrop equilibria provide a good approximation for Nash equilibria in atomic unsplittable congestion games with a large number of small players.
We examine two different definitions of small players. In the first setting, we consider games where each player's weight is small. 
We prove that  when the number of players goes to infinity and their weights to zero, the random flows in all (mixed) Nash equilibria for the finite games converge in distribution to the set of Wardrop equilibria of the corresponding nonatomic limit game. 
In the second setting, we consider an increasing number of players with a unit weight that participate in the game with a decreasingly small probability. 
In this case, the Nash equilibrium flows converge in total variation towards Poisson random variables whose expected values are Wardrop equilibria of a different nonatomic game with suitably-defined costs. The latter can be viewed as symmetric equilibria in a Poisson game in the sense of Myerson, establishing a plausible connection 
between the Wardrop model for routing games and the stochastic fluctuations observed in real traffic. In both settings we provide explicit approximation bounds, and we study the convergence 
of the price of anarchy. Beyond the case of congestion games, we prove a general result on the convergence of large games with random players towards Poisson games. 
%
}%


\MSCCLASS{Primary: 91A14; secondary: 91A07, 91A43, 60B10}
\ORMSCLASS{Primary: Games/group decisions: non-atomic; 
secondary: Networks/graphs: multicommodity}

\maketitle

%


\section{Introduction}
\label{se:intro}

\emph{Nonatomic congestion games} were introduced by \citet{War:PICE1952} as a model for traffic networks with many drivers, where each single agent has a negligible impact on congestion.
The model was stated in terms of continuous flows, which are easier to analyze compared to a discrete model with a finite but large number of players. 
The heuristic justification is that a continuous flow model is a natural approximation to a game with many players, each one having a negligible weight. 
Although this argument is intuitive and plausible, the question of whether nonatomic games are the limit of atomic games has only been formally addressed in special cases, mainly for atomic splittable games with homogeneous players as in \citet{HauMar:N1985}.

Motivated by applications such as road traffic and telecommunications, it is also important to consider the issue of approximating \emph{unsplittable} routing games. In them, players must route a given load over a single path, which can be chosen either deterministically or at random using a mixed strategy. 
We consider the more general class of \emph{atomic unsplittable congestion games} (not necessarily routing games) and allow for heterogeneous players. 
The main question we address is whether Nash equilibria for these games are well approximated by a Wardrop equilibrium of a limiting nonatomic congestion game. 
A convergence result would provide a stronger support for Wardrop's model as an approximation for large games with many small players,
especially if we can estimate the distance between the corresponding equilibrium solutions.

In a \emph{weighted congestion game}, there is a finite number of players who are characterized by a type and a weight. The type determines the set of feasible strategies for the player and the weight determines the player's 
impact on the costs. Moreover, there is a finite set of resources and each strategy corresponds to a subset of these resources.
Players of a given type have the same set of available strategies. A strategy profile for all players induces a flow on each strategy equal to the total weight of players choosing it, as well as a load on each resource equal to the total weight of players using that resource as part of their strategy.
The cost of using a resource is a weakly increasing function of its load, and the cost of a strategy is additive over its resources.
This defines a finite cost-minimization game.
As shown by \citet{Ros:IJGT1973}, every congestion game in which all the players have the same weight admits a potential and has equilibria in pure strategies. 
With heterogeneous weights, only mixed equilibria are guaranteed to exist \citep{Nas:PNAS1950}.

As a special case of congestion games, a routing game features a finite directed network whose edges represent the resources. The origin-destination  pairs encode the types, and the corresponding origin-destination paths provide the strategies.
To illustrate, consider a routing game over a simple network composed of two parallel edges with the same strictly-increasing cost function $c(\argdot)$, as shown in Figure~\ref{fi:parallel-network}, and suppose that there are $n$ players who need to choose an edge to route an identical weight of $w\equiv d/n$. 
Here, $d$ denotes the total weight or demand.
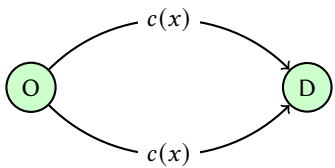
\begin{figure}[ht]
\centering
\begin{tikzpicture}[node distance  = 3 cm,thick,every node/.style={scale=0.8}]
  \tikzset{SourceNode/.style = {draw, circle, fill=green!20}}
  \tikzset{DestNode/.style = {draw, circle, fill=green!20}}
     \node[SourceNode](s){$\mathsf{O}$};
     \node[DestNode,right=of s](t){$\mathsf{D}$};
     \draw[->, bend left = 45](s) to node[fill = white]{$c(x)$} (t);
     \draw[->,bend right = 45 ](s) to node[fill = white]{$c(x)$} (t);
\end{tikzpicture}
\caption{\label{fi:parallel-network} Parallel edge network.}
\end{figure}
In a symmetric equilibrium every player randomizes by choosing each route 
with probability $1/2$. Consequently, the number of players on each route is distributed as a $\Binomial(n,1/2)$ random variable. 
The total load on each edge is therefore $d/n$ times a $\Binomial(n,1/2)$, and converges in distribution to $d/2$, which is precisely the Wardrop equilibrium for a total demand of $d$ units of flow.
Note that already in this simple example there is a multitude of other equilibria where $n_{1}$ and $n_{2}$ players, respectively, choose the upper and lower edges for sure (with $n_{1},n_{2}\le n/2$), and the remaining players (if any) randomize appropriately so as to equalize the expected cost of both routes. {This includes the special case of a pure equilibrium in which half of the players take each route, up to one unit if $n$ is odd.}
As the number of players $n$ tends to infinity, all these different equilibria converge to the unique Wardrop equilibrium with a $(d/2,d/2)$-split of the flow.

In real networks, players are confronted to make decisions while facing multiple sources of uncertainty. 
In particular, even if the population of potential drivers might be known, the subset of drivers that are actually on the road at any given time is random.
In fact, the contribution of an additional car to congestion is small but not negligible, and the congestion experienced by an agent depends basically on how many drivers are on the road at the same time.
To that point, \citet{AngFotLia:AGT2013} and \citet{ComScaSchSti:arXiv2022} studied \emph{Bernoulli congestion games} in which players participate with an independent probability.
If we focus on a small time interval, the probability that any given player participates in the game during that interval emerges naturally as a small parameter.
Motivated by this model, we ask the question whether a congestion game with a large but random number of players and small participation probabilities can be approximated by a nonatomic congestion game.
To address this question we study the convergence of the Nash equilibria for a sequence of Bernoulli games towards the Wardrop equilibria of some nonatomic game.
Taking the limit in this setting yields a different limit game in which the random loads on the resources converge to a family of Poisson random variables, whose expected values can be again characterized as Wardrop equilibria of a suitably defined nonatomic game or, alternatively, as 
an equilibrium of a Poisson game in the sense of \citet{Mye:IJGT1998}. This establishes a novel and precise connection between the Wardrop model and Poisson games.
To provide some insight, consider again the parallel-edge example of Figure~\ref{fi:parallel-network}, except that now each of the $n$ players has a unit weight, but is present in the game with a small probability $d/n$. 
In this case the effective demand is a random variable $D^{n}\sim\Binomial(n,d/n)$, which converges as $n\to\infty$ to a random variable $D\sim\Poisson(d)$. 
Also, in a symmetric equilibrium where each player choses an edge uniformly at random, the load on each route is distributed as a
$\Binomial(n,d/(2n))$ and converges to a $\Poisson(d/2)$, whose expected value is again $d/2$. 
We will show that this convergence holds for all Nash equilibria and for general congestion games with an increasing number of heterogenous players who are active with different vanishing probabilities.

Note that in both cases---vanishing weights and vanishing probabilities---the equilibria of the nonatomic limit games characterize the \emph{expected values} of the random loads. 
However, in the Poisson regime the resource loads remain random in the limit, whereas in the Wardrop 
limit these random loads converge to a constant.
The fact that the Poisson limit retains some variability makes it more suitable to model the traffic flows observed in real networks.
As a matter of fact, real traffic flows exhibit stochastic fluctuations that have 
been empirically confirmed to be close to Poisson distributions, at least under moderate congestion conditions.
As an illustration, Figure~\ref{fig:TrafficCounts} shows the histograms of traffic counts over three consecutive 10 minute intervals, 
observed every Thursday during 2 years on a specific road segment in Dublin. The red curves 
give the expected counts for Poisson distributions with the same mean. These histograms reveal a 
persistent day-to-day variability of traffic flows, so that a model predicting a random distribution is one step closer to reality 
compared to the point estimates provided by Wardrop equilibrium.
Our results establish a theoretically sound connection between the Wardrop and Poisson equilibria, 
which, combined, seem to provide a more sensible model for the traffic flows observed in real networks.

\begin{figure}[!h]
\centering
\includegraphics[scale=.33]{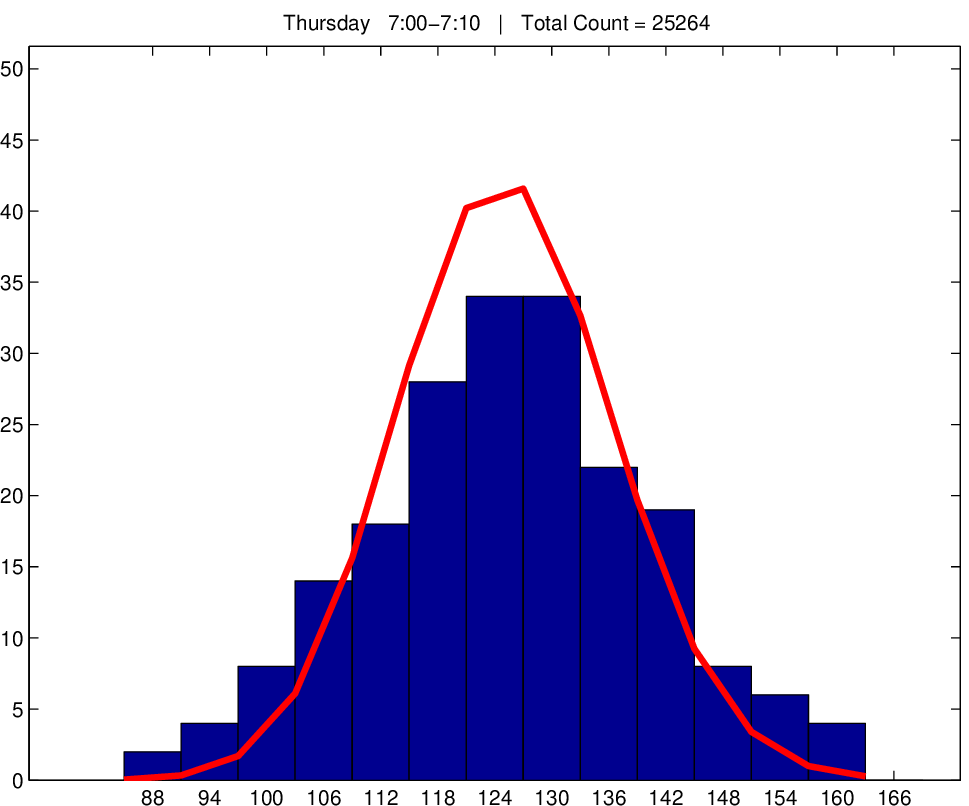}
\includegraphics[scale=.33]{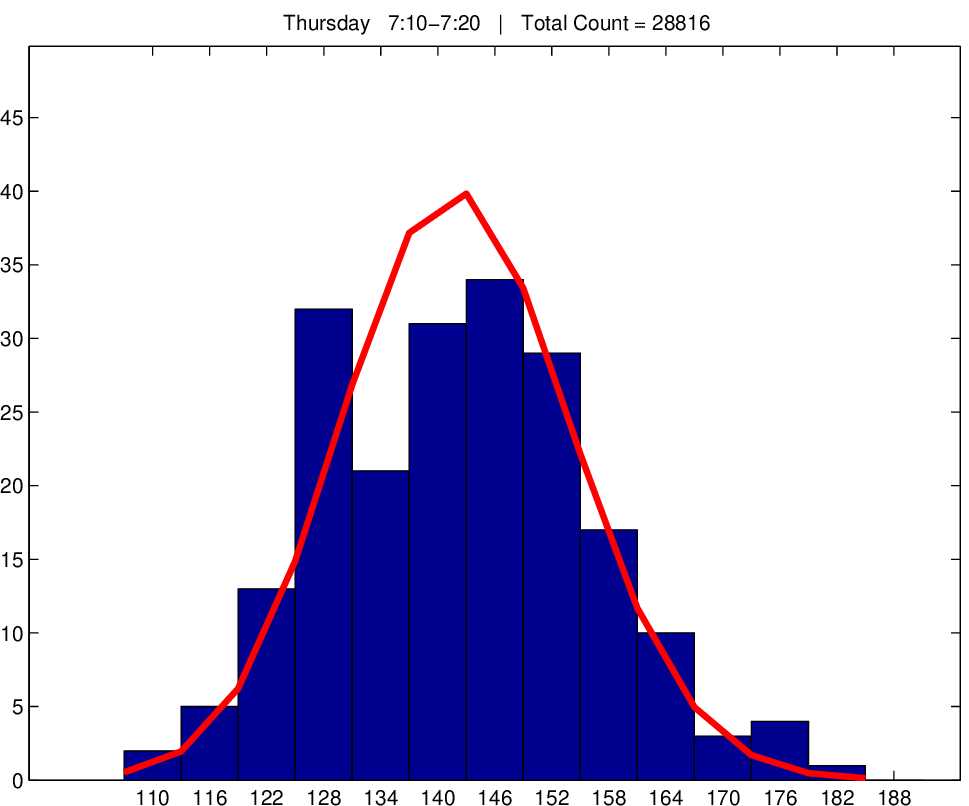}
\includegraphics[scale=.33]{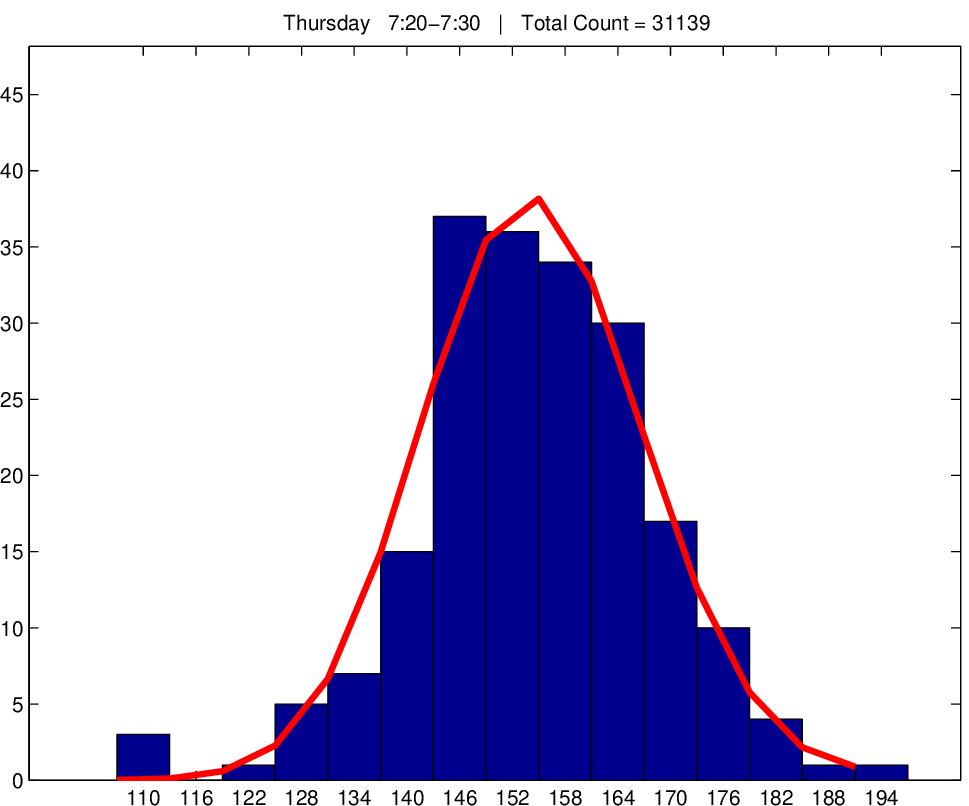}
\caption{\label{fig:TrafficCounts} Traffic counts. Dublin 2017-2018, Thursday 7:00-7:30 AM. Data from Transport Infrastructure Ireland. \url{https://www.tii.ie/roads-tolling/operations-and-maintenance/traffic-count-data/}}
\end{figure}

\subsection{Our Contribution}
\label{suse:our-contribution}

After introducing the relevant classes of congestion games in Section~\ref{se:congestion-games}, Section~\ref{se:convergence-weighted} deals with nonatomic approximations of weighted atomic congestion games.
Under mild and natural conditions we show that Nash equilibria in a sequence of weighted congestion games converge to a Wardrop equilibrium of a limiting nonatomic game. 
More precisely, if the number of players grows to infinity and their weights tend to zero in a way that the aggregate demands converge, then, for any sequence of mixed equilibria in the finite games, the random variables that represent the resource loads converge in distribution to a (deterministic) Wardrop equilibrium for a nonatomic limit game. 
We stress that players are not assumed to be symmetric and that they may have different weights and strategy sets. As long as their weights converge to zero, the random resource loads converge in distribution to some constants, which are precisely a Wardrop equilibrium for the limit game. 
This provides a strong support to Wardrop's model as a sensible approximation for large unsplittable congestion games with many small players.

Section~\ref{se:convergence-stochastic} focuses on approximations of Bernoulli congestion games. {
It considers sequences of Bayesian Nash equilibria for Bernoulli congestion games when the participation probability of players tends to zero, and establishes the convergence of the resource loads towards a family of Poisson variables whose expected values are Wardrop equilibria of another nonatomic game with suitably defined costs. 
In Section~\ref{se:poisson} we unveil the connection between this nonatomic limit game and a Poisson game in the sense of \citet{Mye:IJGT1998}, namely, the Wardrop equilibria are in one-to-one correspondence with the equilibria of an associated Poisson game with countably many players and demands distributed as Poisson random variables. 
We stress that Poisson games were originally introduced axiomatically 
and not as a limit of a sequence of finite games. We close this gap by proving a general result on convergence of sequences of games with Bernoulli players (not necessarily congestion games) towards Poisson games. This constitutes a novel alternative justification for Poisson games.

In Section~\ref{se:rate-of-convergence} we derive non-asymptotic bounds for the distance between 
the  distributions of the  equilibrium loads in the finite games and the Wardrop and Poisson counterparts. These bounds provide explicit estimates for using 
these simpler models  to approximate the equilibria in the finite games.
Finally, in Section~\ref{se:PoA} we turn to investigate the convergence of the price of anarchy (PoA)  and the price of stability (PoS)---as measures of the inefficiency of equilibria---for sequences of weighted congestion games and  Bernoulli congestion games.

In summary, the main contributions of this paper are:
\begin{enumerate}[(a)]
\item
A formal statement and proof of the fact that Nash equilibria of weighted atomic unsplittable congestion games converge to Wardrop equilibria of nonatomic congestion games, as the weights of all players vanish. 
This is achieved under very weak conditions and holds even when weights are different (in that case pure Nash equilibria may not exist).

\item
A formal statement and proof of the fact that, as the participation probability of each player vanishes, the random loads in Bayesian Nash equilibria of Bernoulli congestion games converge to Poisson random variables whose expectations are Wardrop equilibria of a suitably defined nonatomic congestion game.
Again, this is achieved under weak conditions; for instance, the participation probability of different players may differ.

\item
A connection between Poisson and Wardrop equilibria for congestion games, and a formal asymptotic justification of Poisson games as limits of Bernoulli games. 

\item 
Non-asymptotic bounds for the distance between the equilibria in weighted and Bernoulli congestion games to the corresponding Wardrop and Poisson equilibria.

\item
A proof of the convergence of the PoA and PoS for sequences of weighted and Bernoulli congestion games.

\end{enumerate}

\subsection{Related Literature}
\label{suse:related-literature}

Convergence of weighted congestion games towards nonatomic games has been considered before, 
mainly for the case in which players can split their weight over the available strategies.
\citet{HauMar:N1985} were the first to study such convergence issues. 
In a setting of atomic routing games with splittable flows and elastic demands, they proved that when players on each  origin-destination (OD) pair are identically replicated $n$ times, in the limit when $n$ grows, the splittable equilibria converge towards Wardrop equilibria. More recently,
\citet{JacWan:arXiv2018} considered splittable routing on parallel networks with heterogenous players, 
and \citet{JacWan:EJOR2022} studied the approximation of nonatomic aggregative games by a sequence of finite splittable games.
Also the relation between Nash equilibria and a Wardrop-like notion of equilibrium in aggregative games with finitely many players, was studied by \citet{PacGenParKamLyg:IEEETAC2019}.

We are not aware of any results on convergence of equilibria for unsplittable weighted congestion games. 
The closest is \citet{Mil:MOR2000} who studied limits of
\emph{finite crowding games} with an increasing number of players $n$ with identical weights $1/n$, i.e., unweighted congestion games with singleton strategies and type-dependent costs. 
Considering pure equilibria only, he established the convergence of the per-type aggregate strategy loads towards
an equilibrium of a large crowding game.
In another related result, \citet{San:JET2001} proved that infinite potential games can be obtained as limits of finite potential games.
His results are related to ours because every unweighted congestion game is a potential game \citep[][]{Ros:IJGT1973} and, 
conversely, every finite potential game is isomorphic to an unweighted congestion game \citep[][]{MonSha:GEB1996}. 
The difference is that here we consider the more general class of weighted congestion games, which in general do not have a potential structure. 
Moreover, Sandholm studied only the convergence of the potential functions but did not address the convergence of equilibria.
We also mention \citet{FelImmLucRouSyr:STOC2016}, who established asymptotic upper bounds for the PoA in unsplittable games when the number $n$ of players increases and their weight is $1/n$. 
Their results are based on the so-called \emph{$(\lambda,\mu)$-smoothness in the large} which provides upper bounds for the PoA, but again without addressing the convergence of equilibria. 
On a similar vein, \citet{ChedeKKemSch:TEAC2014}
considered nonatomic congestion games as limits of atomic unsplittable games, in order to define altruism and to investigate its impact on the price of anarchy.

In the present paper we consider unsplittable weighted congestion games and present approximation results for the flows themselves, which is a stronger statement and 
therefore technically harder to accomplish. In fact, whereas equilibria for splittable congestion games (as well as pure equilibria in the unsplittable case) are conveniently described in terms of flows and loads that 
live in the same finite dimensional space as Wardrop equilibria, in the unsplittable case with equilibria in mixed strategies these flows and loads become random so that limits must be properly understood 
in terms of convergence of random variables.
This can be used later to derive the convergence of aggregate metrics such as the PoA.

Besides routing games, auctions are another class of games where there is significant previous work on asymptotic properties when the number of players becomes large.
\citet{SatWil:RES1989} and \citet{RusSatWil:E1994} studied the efficiency loss in markets as the number of bidders grows.
\citet{CasdeV:TI2005,LamSho:ICEC2007,BluHol:EC2008,FibGav:IJGT2010}  used asymptotic analysis to understand the convergence of revenue in various auctions settings.
Like our work, several of these references compute equilibria or approximate equilibria as the auctions become larger, and use those characterizations to derive limiting results. Although convergence  might appear natural at first sight, some asymptotic results are fragile and only hold under carefully stated assumptions. For instance, \citet{RobPos:E1976} showed that the gain from deviating from competitive behavior in an exchange economy does not need to diminish as the number of agents in the economy goes to infinity. Also \citet{Swi:E2001} and \citet{FelImmLucRouSyr:STOC2016} showed that the inefficiency in auctions does not need to disappear in the limit with a large number of objects and players.

Several old papers considered the stochastic aspects of traffic, both theoretically and empirically, including the role of the Poisson distribution for modeling it
\citep[see, e.g.,][]{Ada:JICE1936,May:B1954,Oli:OR1961,Bre:BIIS1962,Bre:AMS1963,Buc:TR1967,Mil:TS1970}.
More recently, various authors have studied congestion games with stochastic features,
focusing on the efficiency of equilibria under incomplete information. 
For instance, \citet{GaiMomTie:TCS2008} studied the inefficiency of equilibria for congestion games in which the weight of a player is private information. 
Looking at the cost uncertainty, \citet{NikSti:OR2014} and \citet{PilNikSha:TEAC2016} considered players' risk attitudes in the nonatomic and atomic cases, respectively.
\citet{Rou:ACMTEC2015} showed that whenever player types are independent, the inefficiency bounds for complete information games extend to Bayesian Nash equilibria of the incomplete information game. \citet{WanDoaChe:TRB2014} and \citet{CorHoeSch:TRB2019} looked at similar questions for nonatomic routing games. 
This trend does not only include congestion games: \citet{Sti:Springer2014} studied the efficiency of some classical queueing models on various networks, whereas \citet{HasNowSha:EJOR2018} examined a queueing model with heterogeneous agents and studied how the inefficiency of equilibria varies with the intensity function. 
Closer to our results on Bernoulli games, \citet{AngFotLia:AGT2013} considered congestion games with stochastic players who are risk-averse, restricting their attention to the case of parallel edges.
In the same spirit, \citet{ComScaSchSti:arXiv2022} studied Bernoulli congestion games, where each player $i$ takes part in the game independently with probability $p_{i}$, and found sharp bounds for the PoA as a function of the maximum $p_{i}$.

Games with a random number of players were introduced by 
\citet{Mye:GEB1998,Mye:IJGT1998,Mye:JET2000,Mye:JET2002}, with the main goal of analyzing elections with a large number of voters.
In his seminal paper, \citet{Mye:IJGT1998} showed that the case where the number of players has a Poisson distribution is of particular relevance, being the only case 
where an environmental equivalence holds, i.e., the belief of a player of any type about the type profile of the other active players coincides with the belief of an external game theorist.
\citet{Mye:JET2000} dealt with large Poisson games in which the parameter of the Poisson distribution diverges. 
His approach differs from ours in the sense that he starts with a Poisson distribution, axiomatically justified, and lets the expectation of this distribution go to infinity. 
In our case we start with a finite number of players and let their number diverge in such a way that in the limit we get a Poisson distribution, but not necessarily with a large parameter.
{Our derivation of Poisson games as limits of finite Bernoulli games seems to be new.}
The asymptotic approach we take to Poisson games is based on results from Poisson approximation theory. 
In probability, results about this topic abound. For an overview of this literature, we refer the reader to the books and surveys by
\citet{BarHolJan:OUP1992,BarChe:SUP2005,Nov:PS2019}.

Several other papers have studied the properties of games with random number of players and Poisson games in particular. 
Among them, \citet{Mil:GEB2004} provided a sophisticated analysis of general games with population uncertainty. 
\citet{DeSPim:GEB2009,DeSMerPim:JET2014,MerPim:JET2017} dealt with various properties of equilibria in Poisson games, such as stability and existence of equilibria in undominated strategies.
Other papers apply Poisson games in different settings, not necessarily related to elections.
\citet{LimMat:GEB2009} studied contests with finitely many players, where each player takes part in the game independently with the same probability, whereas
\citet{DuGon:TRB2016} used a Poisson game to model parking congestion and proposed a decentralized and coordinated online parking mechanism to reduce congestion.
Let us finally mention \citet{KorPap:IEEETAC2015} and \citet{BerDes:DGA2017} who studied dynamic games where players arrive at random over time.

\subsection{Organization of the Paper}
\label{suse:organization}

The paper is organized as follows.
Section~\ref{se:congestion-games} describes the different versions of the congestion games that we will study, and introduces the basic notation.
Sections~\ref{se:convergence-weighted} and~\ref{se:convergence-stochastic} study the convergence of sequences of weighted congestion games and Bernoulli congestion games respectively.
Section~\ref{se:poisson} examines Poisson games in greater generality and focuses on convergence to them.
Section~\ref{se:rate-of-convergence} provides the non-asymptotic bounds on the rate of convergence of the finite congestion games to their limits, followed by Section~\ref{se:PoA} where we present convergence results for the PoA and the PoS.
Section~\ref{se:conclusions} presents some conclusions and possible directions for future work.
Appendices~\ref{se:PoA-Proofs} and \ref{se:rate-of-convergence-proof} include auxiliary results and proofs used in Sections~\ref{se:PoA} and \ref{se:rate-of-convergence} respectively, whereas
Appendix~\ref{se:Poisson-Bernoulli} provides a short summary of known results on Poisson approximations for sums of Bernoulli random variables that are used to derive our results.
Finally, Appendix~\ref{se:symbols} contains a glossary of notations for ease of reference.

\section{Congestion Games: Definitions and Notations}
\label{se:congestion-games}

{This section introduces the basic models of congestion games and the variations that will be studied in this paper.}
Informally, a \emph{congestion game} is played over a set of resources whose costs depend on the mass of players using the resource. 
Each player chooses a subset of resources among the allowed subsets, seeking to achieve the minimum possible cost. 
Throughout this paper we consider a finite set of resources $\mathscr{E}$, where each $e\in\mathscr{E}$ is associated with a
weakly increasing continuous cost function $c_{e}:\mathbb{N}_+\to \mathbb{N}_+$. We also fix a finite set of types $\mathscr{T}$ where each type $t\in\mathscr{T}$
is associated with a set of feasible actions $\mathscr{S}_{t}\subseteq 2^{\mathscr{E}}$, which describes the pure strategies. 

To name a standard example of a congestion game, routing games capture the topology of a network structure. 
These games are defined over the set of edges of a finite graph, encoded by $\mathscr{E}$, and the set of OD pairs, encoded by $\mathscr{T}$. 
The set of actions $\mathscr{S}_{t}$ contains the feasible paths for the OD pair representing type $t$, and the costs $c_{e}(\argdot)$ represent the delays when traversing an edge $e$.

The structural objects 
\begin{equation}
\label{eq:structural}
\mathscr{G}=\parens{\mathscr{E},(c_{e})_{e\in\mathscr{E}},\mathscr{T},(\mathscr{S}_{t})_{t\in\mathscr{T}}}
\end{equation} 
will be the same 
in all the congestion games considered hereafter, and the only differences will be in how we describe the set of players and their behavior.
In the nonatomic framework, players are considered to be infinitesimally small and the model is stated in terms of the aggregate mass of players that use each strategy and resource.
In contrast, for weighted congestion games, as well as for Bernoulli congestion games, we have a finite set of players who behave strategically. 
The rest of this section describes these three different models precisely.

As a guide for the notation used in the sequel, we write $\bigtriangleup(\mathscr{S}_{t})$ for the simplex of all probability distributions over the strategy set $\mathscr{S}_{t}$. 
We use capital letters for random variables and lower case for their expected values.
For instance, $X_{e}$ will represent a random load on a resource $e\in\mathscr{E}$ with expected value $x_{e}=\Expect[X_{e}]$,
and we will add a hat $\widehat{x}_{e}$ when referring to an equilibrium.

\subsection{Nonatomic Congestion Games}
\label{suse:nonatomic}

A \emph{nonatomic congestion game} $\Gamma^{\infty}$ is given by a pair $\parens{\mathscr{G},\boldsymbol{d}}$, where $\mathscr{G}$ stands for the structural objects of the game as in \eqref{eq:structural}, and $\boldsymbol{d}=(d_{t})_{t\in\mathscr{T}}$ is a vector of demands with $d_{t}\geq 0$ representing the aggregate demand of type $t$.
The \emph{total demand} is given by the sum over all types $d_{\tot}=\sum_{t\in\mathscr{T}}d_{t}$.

Each demand $d_{t}$ is split over the corresponding strategies $\mathscr{S}_{t}$ and induces loads on the resources. 
Specifically, a \emph{strategy flow} vector $\boldsymbol{y}\coloneqq\parens{y_{t,s}}_{t\in\mathscr{T}, s\in\mathscr{S}_{t}}$ and a \emph{resource load} vector 
$\boldsymbol{x}\coloneqq\parens{x_{e}}_{e\in\mathscr{E}}$ are called \emph{feasible} if they satisfy the following constraints:
\begin{equation}
\label{eq:feasibility-constr}
\forall t\in\mathscr{T},\quad d_{t}=\sum_{s\in\mathscr{S}_{t}}y_{t,s} \text{ with $y_{t,s}\geq 0,\quad$ and }\quad
\forall e\in\mathscr{E}, \quad x_{e}=\sum_{t\in\mathscr{T}}\sum_{s\in\mathscr{S}_{t}}y_{t,s}\mathds{1}_{\braces{e\in s}}.
\end{equation}
The set of such \emph{feasible flow-load pairs} $(\boldsymbol{y},\boldsymbol{x})$ is denoted by $\mathscr{F}(\boldsymbol{d})$. 
Note that the resource loads $\boldsymbol{x}$ are uniquely defined by the strategy flows $\boldsymbol{y}$, but not vice versa. 
Nevertheless, instead of only considering strategy flows, we will refer to flow-load pairs because some concepts are easier to express 
in terms of flows whereas others are defined in terms of loads. 
The notation and nomenclature is inspired by routing games, and most of our examples will be of this type because they are intuitive and well-studied. 
However, we will use these terms in the more general setting of congestion games, even when there is no network structure of resources and strategies.

A \emph{Wardrop equilibrium} is defined as a feasible flow-load pair $(\widehat{\boldsymbol{y}},\widehat{\boldsymbol{x}})\in\mathscr{F}(\boldsymbol{d})$ for which the prevailing cost of all used strategies 
is minimal, or, mathematically,
\begin{equation}
\label{eq:Wardrop}
\forall t\in\mathscr{T},\ \forall s,s'\in\mathscr{S}_{t}\quad
\widehat{y}_{t,s}>0 \implies
\sum_{e\in s}c_{e}(\widehat{x}_{e})\le
\sum_{e\in s'}c_{e}(\widehat{x}_{e}).
\end{equation}
Because only strategies with minimum cost are used, all the strategies used by any given type have the same cost and, as a consequence, any strategy flow decomposition of $\widehat{\boldsymbol{x}}$ yields an equilibrium.
The set of Wardrop equilibria of the nonatomic game $\Gamma^{\infty}$ will be denoted by $\WE(\Gamma^{\infty})$.

\subsection{Weighted Congestion Games}
\label{suse:weighted}

A \emph{weighted congestion game}  is defined as a tuple $\Gamma_{{\scriptscriptstyle \Wg}}=\parens{\mathscr{G},\parens{w_{i},t_{i}}_{i\in\mathscr{N}}}$, 
where $\mathscr{N}$ is a finite set of players and each player $i\in\mathscr{N}$ has a weight $w_{i}\in\mathbb{N}_{+}$ and a type $t_{i}\in\mathscr{T}$
that determines her strategy set $\mathscr{S}_{t_{i}}$. 
Our use of the term ``type'' slightly differs from what is common in game theory, where a type is a random variable associated to each player, 
whose distribution is common knowledge but whose realization is private information. Here, the players' types are deterministic and the type 
of each player is common knowledge.

The aggregate demand for each type $t$ and the total demand are given by 
\begin{equation}
\label{eq:demand-weights}
d_{t}=\sum_{i\colon t_{i}=t}w_{i}, \text{\quad and \quad}
d_{\tot}=\sum_{t\in\mathscr{T}}d_{t}=\sum_{i\in\mathscr{N}}w_{i}.
\end{equation}

Let $\boldsymbol{\sigma}=\parens{\sigma_{i}}_{i\in\mathscr{N}}$ be a \emph{mixed strategy profile}, where $\sigma_{i}\in\bigtriangleup(\mathscr{S}_{t_{i}})$
represents the mixed strategy used by player $i\in\mathscr{N}$, and let 
$\Sigma\coloneqq\times_{i\in\mathscr{N}}\bigtriangleup(\mathscr{S}_{t_{i}})$ be the set of mixed strategy profiles. 
Call $\Expect_{\boldsymbol{\sigma}}$ the expectation with respect to 
the product probability measure $\Prob_{\boldsymbol{\sigma}}\coloneqq\times_{i\in\mathscr{N}}\sigma_{i}$ over the set of pure strategy profiles $\mathscr{S} = \times_{i\in\mathscr{N}}\mathscr{S}_{t_{i}}$.
If $S_{i}$ is the random strategy of player $i$, whose distribution is $\sigma_{i}$, then the probability that 
player $i$ uses a given resource $e$ is
\begin{equation}
\label{eq:prob-e-in-Si}
\sigma_{i,e}\coloneqq\Prob_{\boldsymbol{\sigma}}(e\in S_{i})=
\sum_{s\in\mathscr{S}_{t_{i}}}\sigma_{i}(s)\mathds{1}_{\braces{e\in s}}.
\end{equation}
Accordingly, the strategy flows $Y_{t,s}$ and the resource loads $X_{e}$ become random variables 
\begin{equation}
\label{eq:flow-load-weights-rand}
Y_{t,s}=\sum_{i\colon t_{i}=t}w_{i}\mathds{1}_{\braces{S_{i}=s}} \text{\quad and \quad}
X_{e}=\sum_{i\in\mathscr{N}}w_{i}\mathds{1}_{\braces{e\in S_{i}}},
\end{equation}
given by the random realizations of $S_{i}$. Their expected values are
\begin{equation*}
y_{t,s}\coloneqq\Expect_{\boldsymbol{\sigma}}\bracks{Y_{t,s}}
=\sum_{i\colon t_{i}=t}w_{i}\sigma_{i}(s) \text{\quad and \quad}
x_{e}\coloneqq\Expect_{\boldsymbol{\sigma}}\bracks{X_{e}}=\sum_{i\in\mathscr{N}}w_{i}\sigma_{i,e}.
\end{equation*}
A straightforward calculation shows that the pair $(\boldsymbol{y},\boldsymbol{x})$ satisfies \eqref{eq:feasibility-constr}, so that $\parens{\boldsymbol{y},\boldsymbol{x}}\in\mathscr{F}(\boldsymbol{d})$.

To take the perspective of a fixed player $i$, we assume the player already selected resource $e$ and define the conditional load
\begin{equation}
\label{eq:load-e-i-present}
X_{i,e}\coloneqq w_{i}+ \sum_{j\neq i}w_{j}\mathds{1}_{\braces{e\in S_{j}}}.
\end{equation}
Using this, the expected cost of player $i$ conditional on the player using the resource $e$ is
\begin{equation*}
\Expect_{\boldsymbol{\sigma}}\bracks{c_{e}(X_{e})\mid e\in S_{i}}
=\Expect_{\boldsymbol{\sigma}}\bracks{c_{e}(X_{i,e})}.
\end{equation*}
A strategy profile $\widehat{\boldsymbol{\sigma}}\in\Sigma$ is a \emph{mixed Nash equilibrium} if 
\begin{equation*}
\forall i\in\mathscr{N},\ \forall s,s'\in\mathscr{S}_{t_{i}}\quad
\widehat{\sigma}_{i}(s)>0 \implies 
\sum_{e\in s}\Expect_{\widehat{\boldsymbol{\sigma}}}\bracks{c_{e}(X_{i,e})}\le
\sum_{e\in s'}\Expect_{\widehat{\boldsymbol{\sigma}}}\bracks{c_{e}(X_{i,e})}.
\end{equation*}
The set of mixed Nash equilibria of $\Gamma_{{\scriptscriptstyle \Wg}}$ is denoted by $\MNE(\Gamma_{{\scriptscriptstyle \Wg}})$. 

When all the players have the same weight $w_{i}\equiv w$,
\citet{Ros:IJGT1973} proved that $\Gamma_{{\scriptscriptstyle \Wg}}$ is a potential game and, as a consequence, pure Nash equilibria are guaranteed to exist \citep[see also][]{MonSha:GEB1996}. 
Also, \citet{FotKonSpi:TCS2005} showed that every weighted congestion game with affine costs admits an exact potential. 
Conversely, \citet{HarKliMoh:TCS2011} proved that if $\mathscr{C}$ is a class of cost functions such that every weighted congestion game with costs in $\mathscr{C}$ admits a potential, then $\mathscr{C}$ only contains affine functions.
Existence of pure equilibria in weighted congestion games was further studied by \citet{HarKli:MOR2012}. 
Beyond these cases, one can only guarantee the existence of equilibria in mixed strategies \citep[][]{Nas:PNAS1950}.

\subsection{Bernoulli Congestion Games}
\label{suse:stochastic}

In a weighted congestion game the randomness arises only from the players' mixed strategies. 
In this section, we add another stochastic element: players may not be present in the game.
A \emph{Bernoulli congestion game}  is a congestion game in which each player $i\in\mathscr{N}$ has a unit weight $w_{i}\equiv 1$, but takes part in the game only with some probability $u_{i}\in (0,1]$ and otherwise remains inactive and incurs no cost. The participation events, assumed to be independent, are encoded in random variables $U_{i}\sim \Bernoulli(u_{i})$, which indicate whether player $i\in \mathscr{N}$ is active or not.
A Bernoulli congestion game is denoted by $\Gamma_{{\scriptscriptstyle \Bg}}=\parens*{\mathscr{G},(u_{i},t_{i})_{i\in \mathscr{N}}}$.

The framework is similar to a weighted congestion game in which the $w_{i}$'s are replaced by random weights $U_{i}\in\braces{0,1}$ with expected value $u_{i}$, so that the per-type demands become the random variables
$D_{t}=\sum_{i:t_{i}=t}U_{i}$ with expected values $d_{t}=\Expect\bracks{D_{t}}=\sum_{i:t_{i}=t}u_{i}$.
 The formulas are therefore very similar, with $w_{i}$ replaced by $U_{i}$, or by $u_{i}$ when taking expectations. 
Nevertheless, later on we will see  that the two classes of games behave quite differently in some respects.

Let $\boldsymbol{\sigma}=(\sigma_{i})_{i\in\mathscr{N}}\in\Sigma$ be a profile of mixed strategies. We assume that each player
chooses her mixed strategy before the actual realization of these random variables, so that
no player knows for sure who will be present in the game.
Now randomness is induced both by the random participation and by these mixed strategies.
This is described by a discrete probability space $\parens{\Omega,2^{\Omega},\Prob_{\boldsymbol{\sigma}}}$, where 
$\Omega=\braces{0,1}^{\mathscr{N}}\times\mathscr{S}$ with $\mathscr{S} = \times_{i\in\mathscr{N}}\mathscr{S}_{t_{i}}$ as before,
and
$\Prob_{\boldsymbol{\sigma}}$ is now the probability measure induced by $\boldsymbol{\sigma}$ and by the random participation of players; that is, 
for $\boldsymbol{\omega}\in\braces{0,1}^{\mathscr{N}}$ and $\boldsymbol{s}\in\mathscr{S}$ we have
\begin{equation*}
\Prob_{\boldsymbol{\sigma}}\parens*{\boldsymbol{\omega},\boldsymbol{s}}=\prod_{i\in\mathscr{N}}\Prob_{i}(\omega_{i})\sigma_{i}(s_{i}),
\end{equation*}
with $\Prob_{i}(1)=u_{i}$ and $\Prob_{i}(0)=1-u_{i}$.
The corresponding expectation operator will be denoted $\Expect_{\boldsymbol{\sigma}}$.

As before, $\mathds{1}_{\braces{e\in S_{i}}}$ is a Bernoulli random variable indicating whether the random strategy $S_{i}$ 
includes resource $e$, with \eqref{eq:prob-e-in-Si} still in place.
Additionally, let $U_{i,e}=U_{i}\mathds{1}_{\braces{e\in S_{i}}}$ indicate whether player $i$ is active and chooses resource $e$, for which we have $\Expect_{\boldsymbol{\sigma}}\bracks{U_{i,e}}=u_{i}\,\sigma_{i,e}$.
Then, the total number of active players of type $t\in\mathscr{T}$ using strategy $s\in\mathscr{S}_{t}$, and the total load on resource $e\in\mathscr{E}$, are now the random variables
\begin{equation}
\label{eq:stoch-flow-load}
Y_{t,s}=\sum_{i:t_{i}=t} U_{i} \mathds{1}_{\braces{S_{i}=s}}, \text{\quad and \quad}
X_{e}=\sum_{i\in\mathscr{N}}U_{i,e}.
\end{equation}
The expected strategy flows and resource loads are
\begin{equation}
\label{eq:stoch-expect-flow-load}
y_{t,s}\coloneqq\Expect_{\boldsymbol{\sigma}}\bracks*{Y_{t,s}}=\sum_{i:t_{i}=t}u_{i}\,\sigma_{i}(s), \text{\quad and \quad}
x_{e} \coloneqq \Expect_{\boldsymbol{\sigma}}\bracks*{X_{e}}=\sum_{i\in\mathscr{N}}u_{i}\, \sigma_{i,e}.
\end{equation}
Once again, the pair $(\boldsymbol{y},\boldsymbol{x})$ satisfies \eqref{eq:feasibility-constr}, so that $\parens{\boldsymbol{y},\boldsymbol{x}}\in\mathscr{F}(\boldsymbol{d})$.

When $U_{i,e}=1$, conditional on player $i$ selecting resource $e\in\mathscr{E}$, its load is
$X_{i,e}=1+Z_{i,e}$, where $Z_{i,e}$ represents the number of other players
using that resource, that is,
\begin{equation}
\label{eq:Z}
Z_{i,e}= \sum_{j\neq i}U_{j,e}.
\end{equation}
Then the conditional expected cost for player $i\in\mathscr{N}$ when using this resource is 
\begin{equation}
\label{eq:exp-cost-i}
\Expect_{\boldsymbol{\sigma}}\bracks*{c_{e}\big(X_{e})\mid U_{i,e}=1}= \Expect_{\boldsymbol{\sigma}}\bracks*{c_{e}(X_{i,e})}= \Expect_{\boldsymbol{\sigma}}\bracks*{c_{e}(1+Z_{i,e})}.
\end{equation}
Note that in this setting all players have unit weight so that the resource loads are integer-valued and 
therefore the costs $c_{e}:\mathbb{N}^*\to\mathbb{N}_+$ need only be defined over the positive integers $\mathbb{N}^*=\mathbb{N}\setminus\{0\}$.

A strategy profile $\widehat{\boldsymbol{\sigma}}\in\Sigma$ is a \emph{Bayesian Nash equilibrium}  if
\begin{equation*}
\forall i\in\mathscr{N}, \forall s,s'\in\mathscr{S}_{t_{i}} \quad \widehat{\sigma}_{i}(s)>0 \implies \sum_{e\in s}\Expect_{\widehat{\boldsymbol{\sigma}}}\bracks*{c_{e}(1+Z_{i,e})}\leq\sum_{e\in s'}\Expect_{\widehat{\boldsymbol{\sigma}}}\bracks*{c_{e}(1+Z_{i,e})}.
\end{equation*}
The set of all Bayesian Nash equilibria of $\Gamma_{{\scriptscriptstyle \Bg}}$ is denoted by $\BNE(\Gamma_{{\scriptscriptstyle \Bg}})$.

\begin{remark}
\citet[proposition~3.3]{ComScaSchSti:arXiv2022} showed that $\Gamma_{{\scriptscriptstyle \Bg}}$ is a potential game, which drives the class of Bernoulli congestion games apart from the class of weighted congestion games, which admit a potential only in special cases.
In particular, a Bernoulli congestion game always has equilibria in pure strategies.
However, here we will consider both pure and mixed equilibria.
\end{remark}

\section{Convergence of Weighted Congestion Games}
\label{se:convergence-weighted}

Now that we laid out the games we are considering, we proceed to study the convergence of weighted congestion games to nonatomic congestion games when the number of players increases and their weights decrease. 
The only stochastic element appearing in this section is the fact that players randomize by considering mixed strategies. {
Under mild and natural conditions, we prove that the equilibrium resource loads converge in distribution to constants. These constants are equal to the resource loads prevailing under a Wardrop equilibrium of the corresponding nonatomic game.}

We consider a sequence of weighted congestion games
\begin{equation*}
\Gamma_{{\scriptscriptstyle \Wg}}^{n} = \parens{\mathscr{G},\parens{w_{i}^{n},t_{i}^{n}}_{i\in\mathscr{N}^{n}}}
\end{equation*}
in which the elements that vary over the sequence are the set $\mathscr{N}^{n}$ of players, their weights $w_{i}^{n}$, and their types $t_{i}^{n}$.
We want to study if and how the equilibria for this sequence converge. 
All the notations in Section~\ref{suse:weighted} will remain in place, by simply adding the superscript $n$.
We assume that the number of players goes to infinity and the sequence of weights goes to zero in such a way that the aggregate demand for each type converges. In other words, as there are more players, no player becomes dominant and the demands remain bounded. This is captured by the following asymptotic behavior as $n\to\infty$:
\begin{subequations}
\label{eq:H1}
\begin{align}
\label{eq:players-n}
&\abs{\mathscr{N}^{n}}\to\infty,\\
\label{eq:w-n}
&w^{n}\coloneqq\max_{i\in\mathscr{N}^{n}}w_{i}^{n}\to 0,\\
\label{eq:d-n-t}
&d_{t}^{n}\coloneqq \sum_{i\colon t_{i}^{n}=t}w_{i}^{n}\to d_{t}.
\end{align}
\end{subequations}

Under these conditions, the random loads at equilibrium for the sequence of games $\Gamma_{{\scriptscriptstyle \Wg}}^{n}$ converge in distribution to the loads of a Wardrop equilibrium of the corresponding nonatomic limit game.

\begin{theorem} 
\label{th:T1}
Let $\Gamma_{{\scriptscriptstyle \Wg}}^{n}$ be a sequence of weighted congestion games satisfying the assumptions in \eqref{eq:H1}, and let $\widehat{\boldsymbol{\sigma}}^{n}\in\MNE(\Gamma_{{\scriptscriptstyle \Wg}}^{n})$ be an
arbitrary sequence of mixed Nash equilibria.
Then the corresponding sequence of expected flow-load pairs $\parens{\widehat{\boldsymbol{y}}^{n}, \widehat{\boldsymbol{x}}^{n}}$ is bounded and every accumulation point
$\parens{\widehat{\boldsymbol{y}}, \widehat{\boldsymbol{x}}}$ is a Wardrop equilibrium of the nonatomic congestion game $\Gamma^{\infty}=\parens{\mathscr{G},\parens{d_{t}}_{t\in\mathscr{T}}}$.
Furthermore, along every convergent subsequence, the random flow-load pairs $\parens{{\boldsymbol{Y}}^{n}, {\boldsymbol{X}}^{n}}$ converge in distribution towards $\parens{\widehat{\boldsymbol{y}}, \widehat{\boldsymbol{x}}}$.
\end{theorem}

\proof{Proof.}
As for every mixed strategy profile, the expected flow-load pair $\parens{\widehat{\boldsymbol{y}}^{n},\widehat{\boldsymbol{x}}^{n}}$ belongs to $\mathscr{F}(\boldsymbol{d}^{n})$ 
so that we have \eqref{eq:feasibility-constr}, adding the superscript $n$ to all the terms involved.
Because $d_{t}^{n}\to d_{t}$, it follows from \eqref{eq:feasibility-constr} that the sequence $\parens{\widehat{\boldsymbol{y}}^{n},\widehat{\boldsymbol{x}}^{n}}$ is bounded,
and then passing to the limit in these equations we conclude that every accumulation point $\parens{\widehat{\boldsymbol{y}},\widehat{\boldsymbol{x}}}$ belongs to $\mathscr{F}(\boldsymbol{d})$. 

Now take a convergent subsequence and for simplicity rename it so that the full sequence converges $\parens{\widehat{\boldsymbol{y}}^{n},\widehat{\boldsymbol{x}}^{n}}
\to\parens{\widehat{\boldsymbol{y}},\widehat{\boldsymbol{x}}}$.
From \eqref{eq:flow-load-weights-rand} we get
\begin{equation*}
\Var_{\widehat{\boldsymbol{\sigma}}^{n}}\bracks{X_{e}^{n}}
= \Var_{\widehat{\boldsymbol{\sigma}}^{n}}\bracks*{\sum_{i\in\mathscr{N}^{n}}w_{i}^{n}\mathds{1}_{\braces{e\in S^{n}_{i}}}}
= \sum_{i\in\mathscr{N}^{n}}\parens{w_{i}^{n}}^{2}\widehat{\sigma}_{i,e}^{n}(1-\widehat{\sigma}_{i,e}^{n})
\le\frac{1}{4}\sum_{i\in\mathscr{N}^{n}}\parens{w_{i}^{n}}^{2}.
\end{equation*}
It can be shown similarly that
\begin{equation*}
\Var_{\widehat{\boldsymbol{\sigma}}^{n}}\bracks{Y_{t,s}^{n}}
\le \frac{1}{4}\sum_{i\in\mathscr{N}^{n}}\parens{w_{i}^{n}}^{2}.
\end{equation*}
Conditions \eqref{eq:H1} in turn imply 
\begin{equation*}
\sum_{i\in\mathscr{N}^{n}}(w_{i}^{n})^{2}\leq w^{n}\sum_{i\in\mathscr{N}^{n}}w_{i}^{n}= w^{n}\sum_{t\in\mathscr{T}}d_{t}^{n}\to 0,
\end{equation*}
so that $\Var_{\widehat{\boldsymbol{\sigma}}^{n}}\bracks{Y_{t,s}^{n}}\to 0$ and $\Var_{\widehat{\boldsymbol{\sigma}}^{n}}\bracks{X_{e}^{n}}\to 0$, from which convergence in distribution follows. 

It remains to show that $\parens{\widehat{\boldsymbol{y}},\widehat{\boldsymbol{x}}}$ is a Wardrop equilibrium, that is, we need to establish \eqref{eq:Wardrop}.
If $\widehat{y}_{t,s}>0$, then, for $n$ large enough, we have $\widehat{y}_{t,s}^{n}>0$ and there exists some player $i\in\mathscr{N}^{n}$ of type $t_{i}^{n}=t$ with $\widehat{\sigma}^{n}_{i}(s)>0$.
Note that $i$ actually depends on $n$, so we should write $i^{n}$. For the sake of simplicity we omit the superscript.
The equilibrium condition in $\Gamma_{{\scriptscriptstyle \Wg}}^{n}$ implies that, for each alternative strategy $s'\in\mathscr{S}_{t}$ and for this player $i$,
\begin{equation}
\label{eq:disc}
\sum_{e\in s} \Expect_{\widehat{\boldsymbol{\sigma}}^{n}}\bracks{c_{e}(X_{i,e}^{n})}\le \sum_{e\in s'} \Expect_{\widehat{\boldsymbol{\sigma}}^{n}}\bracks{c_{e}(X_{i,e}^{n})}.
\end{equation}
Because $\abs{X_{i,e}^{n}-X_{e}^{n}}
\leq w_{i}^{n}\leq w^{n}\to 0$, it follows that $X_{i,e}^{n}\to\widehat{x}_{e}$ in distribution. 
Moreover, the loads $X_{e}^{n}\geq 0$ are bounded above by the total demands $d_{\tot}^{n}=\sum_{t\in\mathscr{T}}d_{t}^{n}$, which converge to
$d_{\tot}=\sum_{t\in\mathscr{T}}d_{t}$, so that both $X_{e}^{n}$ and $X_{i,e}^{n}$ are uniformly bounded. 
It follows that
$\Expect_{\widehat{\boldsymbol{\sigma}}^{n}}\bracks{c_{e}(X_{i,e}^{n})}\to c_{e}(\widehat{x}_{e})$ and, letting $n\to\infty$ in \eqref{eq:disc}, we obtain \eqref{eq:Wardrop} as required.
\Halmos
\endproof

\begin{example}
\label{ex:Weathstone-weighted}
We illustrate the previous result on the Wheatstone network shown in Figure~\ref{fi:Wheatstone}. There is a single OD pair and $n\geq 2$ identical players, each one 
with weight $w_{i}\equiv 1/n$ so that the total demand is $d_{\tot}=1$.
\begin{figure}[ht]
\centering
\begin{tikzpicture}[scale=0.8,node distance = 3 cm,thick,every node/.style={scale=0.8}]
  \tikzset{SourceNode/.style = {draw, circle, fill=green!20,minimum size=8mm}}
  \tikzset{DestNode/.style = {draw, circle, fill=green!20,minimum size=8mm}}
  \tikzset{NodeStyle/.style = {draw, circle, fill=blue!20,minimum size=8mm}}
     \node[SourceNode](s) at (-3.5,0) {$\mathsf{O}$};
     \node[NodeStyle](v) at (0,2) {$v$};
     \node[NodeStyle](w) at (0,-2) {$w$};
     \node[DestNode](t) at (3.5,0) {$\mathsf{D}$};
     \draw[->](s) to node[fill = white,sloped]{\small $c_{1}(x)=x$} (v);
     \draw[->](s) to node[fill = white,sloped]{\small $c_{2}(x)=1$} (w);
     \draw[->](v) to node[fill = white,sloped]{\small $c_{4}(x)=1$} (t);
     \draw[->](w) to node[fill = white,sloped]{\small $c_{5}(x)=x$} (t);
     \draw[->](v) to node[fill = white]{\small $c_{3}(x)=0$} (w);
\end{tikzpicture}
\caption{\label{fi:Wheatstone} Wheatstone network.}
\end{figure}
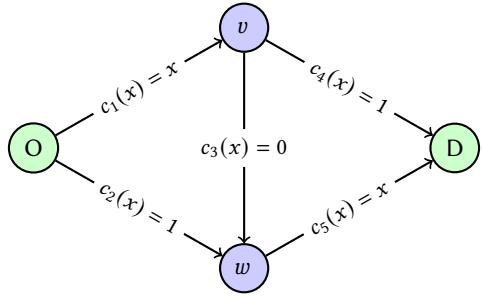
For each $n$, there is a unique symmetric pure Nash equilibrium in which
all players take the zig-zag path $e_{1},e_{3},e_{5}$ and pay a cost equal to $2$.
For $n=2$, we also have a pure equilibrium in which one player takes the upper path whereas the other 
 takes the lower path, and they both pay $3/2$. For $n\geq 3$, the only mixed equilibrium (modulo permutation of the players)
is when $n-1$ players take the zig-zag path and the last player mixes in any possible way over the 
three paths. In the limit when $n\to\infty$, all these equilibria converge to the Wardrop equilibrium of the corresponding 
nonatomic game with a unit flow over the zig-zag path.
\end{example}

For games with strictly-increasing costs $c_{e}(\argdot)$ the Wardrop equilibrium loads $\widehat{x}_{e}$ are unique, 
which yields the following direct consequence.

\begin{corollary} 
\label{co:increasing-costs}
Suppose that the resource costs $c_{e}(\argdot)$ are strictly increasing.
Then, for every sequence $\Gamma_{{\scriptscriptstyle \Wg}}^{n}$ of weighted congestion games satisfying \eqref{eq:feasibility-constr} and each $\widehat{\boldsymbol{\sigma}}^{n}\in\MNE(\Gamma_{{\scriptscriptstyle \Wg}}^{n})$, the random loads $X^{n}_{e}$ converge in distribution to the unique Wardrop equilibrium loads $\widehat{x}_{e}$ 
of the nonatomic limit game $\Gamma^\infty=\parens{\mathscr{G},\parens{d_{t}}_{t\in\mathscr{T}}}$. 
\end{corollary}

As a complement to these results, in Section~\ref{suse:rate-weighted} we will derive explicit bounds on the approximation and rate of convergence under additional conditions on the costs. 
Also, in Section~\ref{se:PoA} we will prove that the PoA for the sequence 
of weighted congestion games $\Gamma_{{\scriptscriptstyle \Wg}}^{n}$ converges to the PoA of the limit game $\Gamma^\infty$.

\section{Convergence of Bernoulli Congestion Games}
\label{se:convergence-stochastic}

In this section we study the convergence of Bernoulli congestion games to nonatomic games when the number of players increases and their participation probabilities decrease. 
We show that the equilibrium flows and loads converge in total variation towards Poisson random variables whose expected values are a Wardrop equilibrium of a corresponding nonatomic game with suitably defined cost functions. 
Nevertheless, we highlight that our results prove also that the distributions  of the random loads converge.

We proceed to study the convergence of Bayesian Nash equilibria $\widehat{\boldsymbol{\sigma}}^{n}\in\BNE(\Gamma_{{\scriptscriptstyle \Bg}}^{n})$ for a sequence of Bernoulli congestion games 
\begin{equation*}
\Gamma_{{\scriptscriptstyle \Bg}}^{n}=\parens*{\mathscr{G},(u_{i}^{n},t_{i}^{n})_{i\in \mathscr{N}^{n}}}.
\end{equation*}
The elements that vary over this sequence are the set $\mathscr{N}^{n}$ of players, their participation probabilities $u_{i}^{n}$, and their types $t_{i}^{n}$, 
as well as the underlying probability spaces encoding the random participation of players and their mixed strategies.
We assume that the number of players goes to infinity while the sequence of participation probabilities goes to zero in such a way that the aggregate expected demand for each type converges. This is captured by the following asymptotic behavior as $n\to\infty$:
\begin{subequations}
\label{eq:H3}
\begin{align}
\label{eq:players-to-infty}
&\abs{\mathscr{N}^{n}}\to\infty,\\
\label{eq:prob-part-to-0}
&u^{n}\coloneqq\max_{i\in\mathscr{N}^{n}}u^{n}_{i}\to 0,\\
\label{eq:dn-to-d}
&{d}^{n}_{t}\coloneqq \sum_{i\colon t^{n}_{i}=t}u^{n}_{i}\to d_{t}.
\end{align}
\end{subequations}

Our analysis is based on results on Poisson approximations for sums of Bernoulli random variables.
{In particular, under \eqref{eq:H3}, the per-type random demands $D_{t}^{n} = \sum_{i\colon t_{i}^{n}=t}U_{i}^{n}$ are known to
converge in total variation to a Poisson variable $D_{t}\sim\Poisson({d}_{t})$ \citep[see, e.g.,][corollary~3.1]{AdeLeu:B2005}.}
We recall that a sequence of probability measures $\mathsf{Q}^{n}$ on $\mathbb{N}$ converges in total variation to the probability measure $\mathsf{Q}$,
whenever $\rho_{\TV}\parens{\mathsf{Q}^{n},\mathsf{Q}}\to 0$, where
\begin{equation}
\label{eq:TV1}
\rho_{\TV}\parens{\mathsf{Q}^{n},\mathsf{Q}}\coloneqq\sup_{A\subset\mathbb{N}}\abs{\mathsf{Q}^{n}(A)-\mathsf{Q}(A)}=\frac{1}{2}\sum_{k\in\mathbb{N}}\abs{\mathsf{Q}^{n}(k)-\mathsf{Q}(k)}.
\end{equation}
Similarly, a sequence of integer-valued random variables $T^{n}$ converges in total variation to $T$ if their distributions satisfy $\rho_{\TV}\parens{\mathscr{L}(T^{n}),\mathscr{L}(T)}\to 0$.
Convergence in total variation is stronger than convergence in distribution---which was the concept used in Section~\ref{se:convergence-weighted}---and is suitable for situations where the limit distribution is discrete.
Appendix~\ref{se:Poisson-Bernoulli} collects further results on Poisson approximations that are used in our subsequent analysis.

\begin{proposition}
\label{pr:Poisson-conv} 
Let $\Gamma_{{\scriptscriptstyle \Bg}}^{n}$ be a sequence of Bernoulli congestion games satisfying the conditions in \eqref{eq:H3}, and $\boldsymbol{\sigma}^{n}$ an arbitrary sequence of mixed strategy profiles.
{Let
$Y^{n}_{t,s}$ and $X^{n}_{e}$ be the random flows and loads given by \eqref{eq:stoch-flow-load}
with $U_{i}=U_{i}^{n}$, $t_{i}=t_{i}^{n}$, and  
$S_{i}=S_{i}^{n}$ drawn independently at random according to $\sigma_{i}^{n}$,
and let $(\boldsymbol{y}^{n},\boldsymbol{x}^{n})$ be the vector of expected flows and loads.} Then:
\begin{enumerate}[\upshape(a)]
\item
\label{it:pr:Poisson-conv-1}
the sequence $(\boldsymbol{y}^{n},\boldsymbol{x}^{n})_{n\in\mathbb{N}}$ is bounded and each accumulation point $(\boldsymbol{y},\boldsymbol{x})$ belongs to $\mathscr{F}({\boldsymbol{d}})$,
\item
\label{it:pr:Poisson-conv-2}
along any convergent subsequence of $(\boldsymbol{y}^{n},\boldsymbol{x}^{n})$, the random flows $Y^{n}_{t,s}$ and loads $X^{n}_{e}$ converge in total variation to $Y_{t,s}\sim \Poisson(y_{t,s})$ and $X_{e}\sim \Poisson(x_{e})$ respectively, and
\item
\label{it:pr:Poisson-conv-3}
the limit Poisson variables $(Y_{t,s})_{t\in\mathscr{T},s\in\mathscr{S}_{t}}$ are independent.
\end{enumerate}
\end{proposition}

\proof{Proof.}
\ref{it:pr:Poisson-conv-1}
Because the expected demands ${d}_{t}^{n}$ are convergent, this follows directly from \eqref{eq:stoch-expect-flow-load}, which
in this case reads
\begin{equation*}
{d}_{t}^{n}=\sum_{s\in\mathscr{S}_{t}}y^{n}_{t,s}, \text{\quad and \quad}
x^{n}_{e} =\sum_{t\in\mathscr{T}}\sum_{s\in\mathscr{S}_{t}}y^{n}_{t,s}\mathds{1}_{\braces{e\in s}}.
\end{equation*}

\medskip
\noindent
\ref{it:pr:Poisson-conv-2}
Take a convergent subsequence and for simplicity rename it to be the full sequence $(\boldsymbol{y}^{n},\boldsymbol{x}^{n})\to(\boldsymbol{y},\boldsymbol{x})$. 
Considering a Poisson variable $V_{e}^{n}\sim\Poisson(x_{e}^{n})$ and using Theorem~\ref{th:Barbour}\ref{it:th:Barbour-1} we have
\begin{equation*}
\rho_{\TV}\parens*{\mathscr{L}(X_{e}^{n}),\mathscr{L}(V_{e}^{n})}\leq u^{n},
\end{equation*}
so that, using the triangle inequality and \eqref{eq:TV-bound}, we conclude
\begin{equation}
\label{eq:rate}
\rho_{\TV}\parens*{\mathscr{L}(X_{e}^{n}),\mathscr{L}(X_{e})}\leq u^{n}
+\abs{x_{e}^{n}- x_{e}}\to 0.
\end{equation}
A similar argument shows that $Y_{t,s}^{n}$ converges in total variation to $Y_{t,s}$.

\medskip
\noindent
\ref{it:pr:Poisson-conv-3}
Consider the joint moment generating function $M^{n}$ of the random variables $Y_{t,s}^{n}$
\begin{equation*}
M^{n}(\boldsymbol{\lambda})=\Expect_{\boldsymbol{\sigma}^{n}}\bracks*{\exp\parens*{\sum_{t\in\mathscr{T}}\sum_{s\in\mathscr{S}_{t}}\lambda_{t,s}Y^{n}_{t,s}}},
\end{equation*}
where $\boldsymbol{\lambda}=(\lambda_{t,s})_{t\in\mathscr{T},s\in\mathscr{S}_{t}}$. 
Using the fact that 
\begin{equation*}
Y^{n}_{t,s}=\sum_{i\colon t_{i}^{n}=t}U^{n}_{i}\mathds{1}_{\braces{S_{i}^{n}=s}},
\end{equation*}
with $U^{n}_{i}\mathds{1}_{\braces{s_{i}^{n}=s}}$ independent across players and types (although not across strategies), we have
\begin{align*}
M^{n}(\boldsymbol{\lambda})
&=\Expect_{\boldsymbol{\sigma}^{n}}\bracks*{\exp\parens*{\sum_{t\in\mathscr{T}}\sum_{s\in\mathscr{S}_{t}}\sum_{i\colon t_{i}^{n}=t}\lambda_{t,s}U_{i}^{n}\mathds{1}_{\braces{S_{i}^{n}=s}}}}\\
&=\prod_{t\in\mathscr{T}}\prod_{i\colon t_{i}^{n}=t}\Expect_{\boldsymbol{\sigma}^{n}}\bracks*{\exp\parens*{U_{i}^{n}\sum_{s\in\mathscr{S}_{t}}\lambda_{t,s}\mathds{1}_{\braces{S_{i}^{n}=s}}}}\\
&=\prod_{t\in\mathscr{T}}\prod_{i\colon t_{i}^{n}=t}\parens*{1-u_{i}^{n}+u_{i}^{n}\,\Expect_{\boldsymbol{\sigma}^{n}}\bracks*{\exp\parens*{\sum_{s\in\mathscr{S}_{t}}\lambda_{t,s} \mathds{1}_{\braces{S_{i}^{n}=s}}}}}\\
&=\prod_{t\in\mathscr{T}}\prod_{i\colon t_{i}^{n}=t}\parens*{1-u_{i}^{n}+u_{i}^{n}\sum_{s\in\mathscr{S}_{t}}\sigma^{n}_{i}(s)\exp\parens*{\lambda_{t,s}}}.
\end{align*}
Taking logarithms and using the fact that $\ln(1+y)=y+O(y^{2})$, it follows that 
\begin{align*}
\ln M^{n}(\boldsymbol{\lambda})
&=\sum_{t\in\mathscr{T}}\sum_{i\colon t_{i}^{n}=t}\ln\parens*{1+u_{i}^{n}\sum_{s\in\mathscr{S}_{t}}\sigma^{n}_{i}(s)\bracks*{\exp\parens*{\lambda_{t,s}}-1}}\\
&=\sum_{t\in\mathscr{T}}\sum_{i\colon t_{i}^{n}=t}\parens*{\sum_{s\in\mathscr{S}_{t}}u_{i}^{n}\,\sigma^{n}_{i}(s)\bracks*{\exp\left(\lambda_{t,s}\right)-1}+O((u_{i}^{n})^{2})}\\
&=\sum_{t\in\mathscr{T}}\sum_{s\in\mathscr{S}_{t}}y^{n}_{t,s}\bracks*{\exp\parens*{\lambda_{t,s}}-1}+\sum_{i\in\mathscr{N}^{n}}O((u_{i}^{n})^{2}).
\end{align*}
In view of the conditions in \eqref{eq:H3} the sum $\sum_{i\in\mathscr{N}^{n}}O((u_{i}^{n})^{2})$ converges to zero, and therefore
\begin{equation*}
\lim_{n\to\infty}M^{n}(\boldsymbol{\lambda})=\prod_{t\in\mathscr{T}}\prod_{s\in\mathscr{S}_{t}}\exp\parens*{y_{t,s}\bracks*{\exp(\lambda_{t,s})-1}},
\end{equation*}
which is the moment generating function of a family of independent Poisson random variables with parameters $y_{t,s}$.
\Halmos
\endproof

\begin{remark}
Even if the strategy flows $Y_{t,s}$ are independent, in general the resource loads $X_{e}$ are not. 
For instance in a trivial network with a single OD pair connected by a unique path, all the edges will carry exactly the same load.
\end{remark}

With these preliminary facts, we proceed to introduce the nonatomic congestion game  $\breve{\Gamma}^{\infty}$
 that characterizes the limit of the sequence $\Gamma_{{\scriptscriptstyle \Bg}}^{n}$.
Recall that the costs $c_{e}:\mathbb{N}^*\to\mathbb{N}_+$ in Bernoulli games are only defined over the positive integers.
The game $\breve{\Gamma}^{\infty}=\parens{\breve{\mathscr{G}},({d}_{t})_{t\in\mathscr{T}}}$ 
is given by the limiting demands ${d}_{t}$ and the auxiliary costs $\breve{c}_{e}:\mathbb{N}_+\to\mathbb{N}_+\cup\{+\infty\}$ 
where $\breve{c}_{e}(x)$ is
defined for $x\geq 0$ by taking a random variable $X\sim\Poisson(x)$ and setting 

\begin{equation}
\label{eq:cost-Poisson}
\breve{c}_{e}(x)\coloneqq\Expect\bracks{c_{e}(1+X)}.
\end{equation}

Notice that, similar to \eqref{eq:exp-cost-i}, the $+1$ is due to the fact that a player must account for her own presence on the resource.
To ensure that these expected costs are well defined and smooth, we impose the following mild condition:
\begin{equation}
\label{eq:HP}
\exists \nu\in\mathbb{N}\text{ and }d_{\max} > d_{\tot} \text{ with }\Expect\bracks*{\abs{\Delta^{2} c_{e}(1+V)}}\leq\nu\text{ for all } e\in\mathscr{E} \text{ and } V\sim\Poisson(d_{\max}),
\end{equation}
where $\Delta^{2} c_{e}(k)=c_{e}(k+2)-2c_{e}(k+1)+c_{e}(k)$.
This condition holds in particular for costs with subexponential growth $c_{e}(k)\leq b \exp(a k)$ for some $a,b\geq 0$, which include all polynomials.
The condition fails for fast growing functions such as $c_{e}(k)= k!$ or $c_{e}(k)=\exp(\exp(k))$.
The following result summarizes the main consequences of \eqref{eq:HP}.

\begin{lemma}
\label{le:aux}
Assume \eqref{eq:HP} and let 
\begin{align}
\label{eq:zeta}
\zeta&\coloneqq (\mathrm{e}^{d_{\max}}-1)\nu+\max_{e\in\mathscr{E}}\,(c_{e}(2)- c_{e}(1)),\\
\label{eq:Lambda}
\Lambda(u)&\coloneqq\frac{\nu\,d_{\max}}{2}\frac{u \mathrm{e}^{u}}{(1-u)^{2}}+u\,\zeta.
\end{align}
Then, the costs $\breve c_{e}(\argdot)$ are finite and of class $C^{2}$ on $[0,d_{\max}]$ with 
$0\leq \breve{c}_{e}'(x)\leq\zeta$ for all $x\in[0,d_{\max}]$.
Moreover, let $\Gamma_{{\scriptscriptstyle \Bg}}$ be a Bernoulli congestion game with $d_{\tot}\leq d_{\max}$ and let $u=\max_{i\in\mathscr{N}}u_{i}$.
Let $X_{e}$ be the random loads in a mixed strategy profile $\boldsymbol{\sigma}$, and 
$Z_{i,e}=X_{e}-U_{i,e}$ be the loads excluding player $i$. Then,
the expected values
$x_{e}=\Expect_{\boldsymbol{\sigma}}\bracks{X_{e}}$ and $z_{i,e}=\Expect_{\boldsymbol{\sigma}}\bracks{Z_{i,e}}$
satisfy $\abs{z_{i,e}-x_{e}}=\Expect_{\boldsymbol{\sigma}}\bracks{U_{i,e}}\leq u$,
and we have
\begin{equation}
\label{eq:cP-c1}
\abs*{\Expect_{\widehat{\boldsymbol{\sigma}}}\bracks*{c_{e}(1+Z_{i,e})}-\breve{c}_{e}(x_{e})}\leq \Lambda(u).
\end{equation}
\end{lemma}

\proof{Proof.}
The smoothness and the bound $0\leq\breve{c}_{e}'(x)\leq\zeta$ for the auxiliary costs 
follow directly from Proposition~\ref{pr:Poisson-Delta-j}\ref{it:pr:Poisson-Delta-j-2} and Corollary~\ref{co:Poisson-alpha}\ref{co:Poisson-alpha-a}. 
In particular, $\breve{c}_{e}(\argdot)$ is $\zeta$-Lipschitz,
and then \eqref{eq:cP-c1} follows by using a triangle inequality and Theorem~\ref{th:Barbour}\ref{it:th:Barbour-2}.
\Halmos
\endproof

With all these preliminaries we are ready to prove the convergence of equilibria for Bernoulli congestion games.

\begin{theorem}
\label{th:Bayes-Poisson} 
Let $\Gamma_{{\scriptscriptstyle \Bg}}^{n}$ be a sequence of Bernoulli congestion games satisfying \eqref{eq:H3} and \eqref{eq:HP}, and let $\widehat{\boldsymbol{\sigma}}^{n}\in\BNE(\Gamma_{{\scriptscriptstyle \Bg}}^{n})$ be an
arbitrary sequence of Bayesian Nash equilibria. Then the corresponding sequence of expected flow-load pairs $(\widehat{\boldsymbol{y}}^{n},\widehat{\boldsymbol{x}}^{n})$ 
is bounded and every accumulation point $(\widehat{\boldsymbol{y}},\widehat{\boldsymbol{x}})$ is a Wardrop equilibrium for the nonatomic congestion game $\breve{\Gamma}^{\infty}$.
Furthermore, along every such convergent subsequence the random flows $Y^{n}_{t,s}$ and loads $X^{n}_{e}$ converge in total variation
 to Poisson limits $Y_{t,s}\sim \Poisson(\widehat{y}_{t,s})$ and $X_{e}\sim \Poisson(\widehat{x}_{e})$, where the variables $Y_{t,s}$
 are independent. 
 \end{theorem}
 
\proof{Proof.}
From Proposition~\ref{pr:Poisson-conv} we have that $(\widehat{\boldsymbol{y}}^{n},\widehat{\boldsymbol{x}}^{n})$ is bounded and also that along any convergent subsequence
$Y^{n}_{t,s}$ and $X^{n}_{e}$ converge in total variation to Poisson limits $Y_{t,s}\sim \Poisson(\widehat{y}_{t,s})$ 
and $X_{e}\sim \Poisson(\widehat{x}_{e})$ with $Y_{t,s}$ independent. 
It remains to show that every accumulation point is a Wardrop equilibrium for $\breve{\Gamma}^{\infty}$.
Take a convergent subsequence and rename it to be the full 
sequence $(\widehat{\boldsymbol{y}}^{n},\widehat{\boldsymbol{x}}^{n})\to(\widehat{\boldsymbol{y}},\widehat{\boldsymbol{x}})$. 
We must prove that, for each type $t$ and each pair of strategies $s,s'$ in $\mathscr{S}_{t}$, we have
\begin{equation}
\label{eq:wardrop2}
\widehat{y}_{t,s}>0\quad \implies \quad\sum_{e\in s} \breve c_{e}(\widehat{x}_{e}) \leq \sum_{e\in s'}\breve c_{e}(\widehat{x}_{e}).
\end{equation}
A strict inequality $\widehat{y}_{t,s}>0$ implies that for all $n$ large enough we have $\widehat{y}^{n}_{t,s}>0$ and there must be a player $i=i^{n}$ of type $t_{i}^{n}=t$ with $\widehat{\sigma}_{i}^{n}(s)>0$.
The equilibrium condition in $\Gamma_{{\scriptscriptstyle \Bg}}^{n}$ implies that for each alternative strategy $s'\in\mathscr{S}_{t}$ for player $i$ we have 
\begin{equation}
\label{eq:disc2}
\sum_{e\in s} \Expect_{\widehat{\boldsymbol{\sigma}}^{n}}\bracks{c_{e}(1+Z_{i,e}^{n})}
\leq \sum_{e\in s'} \Expect_{\widehat{\boldsymbol{\sigma}}^{n}}\bracks{c_{e}(1+Z_{i,e}^{n})}.
\end{equation}
Using Lemma~\ref{le:aux}, we get
\begin{equation*}
\abs*{\Expect_{\widehat{\boldsymbol{\sigma}}^{n}}\bracks*{c_{e}(1+Z^{n}_{i,e})}-\breve{c}_{e}(x_{e}^{n})}\leq \Lambda(u^{n}) \to 0,
\end{equation*}
and, because $x_{e}^{n}\to\widehat{x}_{e}$, it follows that
$\Expect_{\widehat{\boldsymbol{\sigma}}^{n}}\bracks{c_{e}(1+Z^{n}_{i,e})}\to \breve{c}_{e}(\widehat{x}_{e})$.
Thus, letting $n\to\infty$ in \eqref{eq:disc2}, we obtain \eqref{eq:wardrop2} as required. 
\Halmos
\endproof

\begin{example}
\label{ex:sequence-stoch-games}
Consider a sequence $\parens{\Gamma_{{\scriptscriptstyle \Bg}}^{n}}_{n\in\mathbb{N}}$ of games played on the Wheatstone network in Figure~\ref{fi:Wheatstone}, 
with $n$ identical players and $u_{i}\equiv u=1/n$ so that the total expected demand is $d_{\tot}=1$. Notice that the conditional expected cost of a player $i\in\mathscr{N}$ when using $e_{1}$ or $e_{5}$ equals $1+u(x-1)$, where $x$ is the number of players using the edge.
It follows that the zig-zag path $(e_{1}$, $e_{3}$, $e_{5})$ is strictly dominated by a linear combination of the upper and lower paths, and no player will use the zig-zag path. 
This is in sharp contrast with the weighted congestion games described in Example~\ref{ex:Weathstone-weighted}, where all but one players were choosing the zig-zag path. 
This difference is explained by the fact that each player has a unit weight and has a significant impact on the costs. 

For an even number of players $n=2 k$, there is a unique PNE (modulo permutations of players), where half of the players choose the upper path and the other 
half take the lower path. 
In this equilibrium, the expected cost for each player is $(2.5n-1)/n^2$. 
When $n=2 k+1$ is odd, there is a NE where each edge gets $k$ players for sure, and the last player randomizes arbitrarily between these two paths.
For all $n$, there is also a symmetric mixed equilibrium where every player mixes with $q=1/2$ between the upper and lower paths,
and the expected cost of each player is $(5n-1)/(2n^2)$. 
In general, one can prove that every mixed Nash equilibrium is of the form where $k_{1}$ players take the upper path, $k_{2}$ players take the lower path, and the remaining $k_{3}=n-k_{1}-k_{2}$ players, 
with $k_{3}-1>\abs{k_{2}-k_{1}}$, randomize between these two paths with exactly the same probabilities $(q,1-q)$ with $q=\frac{1}{2}\parens*{1+\frac{k_{2}- k_{1}}{k_{3}-1}}$, so as to equalize their costs. 
As $n\to\infty$, all these equilibria converge to the Wardrop equilibrium of a nonatomic game where the costs of $e_{1}$ and $e_{5}$ 
become $\breve{c}(x)=1+x$. In this Wardrop equilibrium, half of the demand goes on the upper path and the remaining half on the lower path.
\end{example}

Under a mild additional condition, the equilibrium loads of the nonatomic limit game are unique and every sequence of equilibria converges.
The following is a direct consequence of Theorem~\ref{th:Bayes-Poisson}
and the strict monotonicity of $\breve{c}_{e}(\argdot)$ established in Corollary~\ref{co:Poisson-alpha}\ref{co:Poisson-alpha-b}.
\begin{corollary}\label{cor:Bayes-Poisson}
Let $c_{e}:\mathbb{N}^*\to\mathbb{N}$ be weakly increasing and nonconstant, and assume that the conditions \eqref{eq:HP} hold.
Then, the extended functions 
$\breve{c}_{e}(\argdot)$ are strictly increasing and the loads $\widehat{x}_{e}$ are the same in every Wardrop equilibrium for $\breve{\Gamma}^{\infty}$. 
Moreover, for every sequence $\Gamma_{{\scriptscriptstyle \Bg}}^{n}$ of Bernoulli congestion games satisfying the conditions in \eqref{eq:H3} and every sequence $\widehat{\boldsymbol{\sigma}}^{n}\in\BNE(\Gamma_{{\scriptscriptstyle \Bg}}^{n})$, the random loads
$X_{e}^{n}$ converge in total variation to a random variable $X_{e}\sim\Poisson(\widehat{x}_{e})$.
\end{corollary}

Theorem~\ref{th:Bayes-Poisson} and its Corollary~\ref{cor:Bayes-Poisson} provide one of our main results. 
They show that the equilibrium flows in Bernoulli congestion games converge in total variation towards 
Poisson variables whose expected values coincide with the Wardrop equilibria of an auxiliary nonatomic game $\breve{\Gamma}^\infty$.
In Section~\ref{se:poisson} we supplement this result by showing that the limit game $\breve{\Gamma}^\infty$ can be interpreted 
as a Poisson game in the sense of \citet{Mye:IJGT1998}, and that the equilibria of such Poisson games can similarly
be determined by means of a Wardrop equilibrium of $\breve{\Gamma}^\infty$. This provides a 
novel justification for Poisson games, as well as an alternative interpretation of Wardrop equilibria.

In addition to these results, Section~\ref{suse:rate-Bernoulli} presents 
explicit bounds for the approximation and rate of convergence of the distributions of the resource loads, and Section~\ref{se:PoA} 
shows that the PoA for the sequence of Bernoulli congestion games $\Gamma_{{\scriptscriptstyle \Bg}}^{n}$ converges to the 
PoA in the nonatomic limit $\breve{\Gamma}^{\infty}$.

\section{Poisson Games}\label{se:poisson}

{In this section we expose the connection between the nonatomic game $\breve{\Gamma}^\infty$ obtained in Section~\ref{se:convergence-stochastic}
 as a limit of Bernoulli congestion games, and Poisson games in the sense of \citet{Mye:IJGT1998}. We first formally introduce Myerson's framework and Poisson games, 
and conclude that Poisson congestion games arise naturally as a limit of Bernoulli congestion games (which was essentially proved in Theorem~\ref{th:Bayes-Poisson}), and 
that the equilibria of the Poisson game coincide with the Wardrop equilibria of the nonatomic game. This is a novel justification for Wardrop's model and shows that equilibria in Poisson congestion games can be calculated from the solution of a related Wardrop equilibrium. Later on, we consider more general games with bounded costs and establish the convergence of sequences of games with Bernoulli players towards Poisson games, providing a microfoundation for Poisson games.}

\subsection{Games with Population Uncertainty} 

\emph{Games with population uncertainty} were introduced by \citet{Mye:IJGT1998} as a model for situations in which players do not know exactly 
which players participate in the game, but share an awareness of the number and type of players that might be present. These games are 
described by a finite set of types $\mathscr{T}$ and  a joint probability distribution $\mu$ over $\mathbb{N}^{\mathscr{T}}$ that characterizes the random 
number of players of each type $\boldsymbol{N} = (N_{t})_{t\in\mathscr{T}}$ that take part in the game.\footnote{To avoid trivialities we suppose that for each $t\in\mathscr{T}$ the random variable $N_{t}$ is not identically zero.} Each player of type
$t\in\mathscr{T}$ has a corresponding strategy set $\mathscr{S}_{t}$, and a cost $C(\boldsymbol{Y}^{-t};t,s)$ that depends on the chosen
action $s\in\mathscr{S}_{t}$ and on the strategy flows $\boldsymbol{Y}^{-t}$ induced by the
actions of the \emph{other players} as described below. Altogether, a game with population uncertainty is defined by the tuple $\Gamma_{{\scriptscriptstyle \Pg}}=(\mathscr{T},\mu,(\mathscr{S}_{t})_{t\in\mathscr{T}},C)$.

Because players cannot distinguish the specific identities of the other players, and only know the joint probability distribution $\mu$ according to which 
 $\boldsymbol{N} = (N_{t})_{t\in\mathscr{T}}$ is drawn, they are treated symmetrically by assuming that all players of type $t$ adopt the same mixed strategy $\sigma_{t}\in\bigtriangleup(\mathscr{S}_{t})$,
which only depends on their type. The measure $\mu$ and the strategy profile $\boldsymbol{\sigma}=\parens{\sigma_{t}}_{t\in\mathscr{T}}$ determine the distribution of the strategy flows 
$\boldsymbol{Y}= (Y_{t,s})_{t\in\mathscr{T},s\in\mathscr{S}_{t}}$, where $Y_{t,s}$ is
the number of players of type $t$ that choose the action $s$. Indeed, conditionally on the event $N_{t}=\bar{n}_{t}$, the flows $Y_{t}=(Y_{t,s})_{s\in\mathscr{S}_{t}}$ for type $t\in\mathscr{T}$ are distributed across strategies as independent multinomials $Y_{t}\sim \Multi(\bar{n}_{t},\sigma_{t})$.
Hence, for each $\boldsymbol{n}=\parens{n_{t,s}}_{t\in\mathscr{T},s\in\mathscr{S}_{t}}$ the distribution of $\boldsymbol{Y}$ is given by
\begin{equation}
\label{eq:Udist}
\Prob_{\mu,\boldsymbol{\sigma}}(\boldsymbol{Y}=\boldsymbol{n})=\mu(\bar{\boldsymbol{n}})\prod_{t\in\mathscr{T}}\parens*{\bar{n}_{t} !\prod_{s\in\mathscr{S}_{t}}\frac{\sigma_{t}(s)^{n_{t,s}}}{n_{t,s}!}},
\end{equation}
where $\bar{\boldsymbol{n}}\coloneqq\parens{\bar n_{t}}_{t\in\mathscr{T}}$, with $\bar{n}_{t}\coloneqq\sum_{s\in\mathscr{S}_{t}}n_{t,s}$.

Now, in order to determine the expected cost for a generic player of type $t$ who is active in the game, we need the conditional 
distribution assessed by such a player for the strategy flows $\boldsymbol{Y}^{-t}$ induced by the other players.
To this end, 
let $\boldsymbol{N}^{-t}=(\boldsymbol{N}^{-t}_{t'})_{t'\in\mathscr{T}}$ be the random vector giving the number of players for each type, excluding this active player
of type $t$.
As argued by \citet{Mye:GEB1998,Mye:IJGT1998} \citep[see also][]{Mil:GEB2004}, the \emph{posterior distribution of $\boldsymbol{N}^{-t}$} can be identified with
the Palm distribution of $\mu$ viewed as a finite point process, that is}
\begin{equation*}
\mu(\bar{\boldsymbol{n}}\mid t)
\coloneqq\Prob_{\mu,\boldsymbol{\sigma}}(\boldsymbol{N}^{-t}=\bar{\boldsymbol{n}})
=\frac{(\bar{n}_{t}+1)\,\mu(\bar{\boldsymbol{n}}+\boldsymbol{\delta}_{t})}{\Expect_{\mu}\bracks{N_{t}}},
\end{equation*}
where $\bar{\boldsymbol{n}}+\boldsymbol{\delta}_{t}$ denotes the vector $\bar{\boldsymbol{n}}$ with $\bar{n}_{t}$ replaced by $\bar{n}_{t}+1$. 
Accordingly, the posterior distribution of the strategy flows $\boldsymbol{Y}^{-t}=(\boldsymbol{Y}^{-t}_{t',s})_{t'\in\mathscr{T},s\in\mathscr{S}_{t'}}$ induced by the remaining players is given by \eqref{eq:Udist}, with $\mu(\bar{\boldsymbol{n}})$ replaced by $\mu(\bar{\boldsymbol{n}}\mid t)$.
The expected cost for a player of type $t$ is computed according to this posterior distribution, and
an \emph{equilibrium} is then defined as a family 
of type-dependent mixed strategies $\widehat{\boldsymbol{\sigma}}=(\widehat{\sigma}_{t})_{t\in\mathscr{T}}$, with $\widehat{\sigma}_{t}\in\bigtriangleup(\mathscr{S}_{t})$ 
such that
\begin{equation}
\label{eq:Poisson-eq}
\forall t\in\mathscr{T}, \forall s,s'\in\mathscr{S}_{t}\quad \widehat{\sigma}_{t}(s)>0
\implies \Expect_{\mu,\widehat{\boldsymbol{\sigma}}}\bracks{C( Y^{-t};t,s)}
\leq \Expect_{\mu,\widehat{\boldsymbol{\sigma}}}\bracks{C(Y^{-t};t,s')}.
\end{equation}

In order for the expectations to be well defined, \citet{Mye:IJGT1998} assumed the functions $C(\argdot;t,s)$ to be bounded. 
For congestion games this would require the costs $c_{e}(\argdot)$ to be bounded.
Fortunately, this can be relaxed and we only require the much weaker condition in \eqref{eq:HP}. 
Later on, when considering more general Poisson games, we will go back to Myerson's assumptions.

\begin{remark}
The Bernoulli congestion games in Section~\ref{suse:stochastic} fall in the framework of games with population uncertainty 
where $N_{t}=\sum_{i:t_{i}=t}U_{i}$ is the sum of independent non-homogeneous Bernoulli random variables. 
However, we considered not only strategies defined by the player's type, but also
 asymmetric equilibria in which players choose their strategies individually. 
\end{remark}

\begin{example}
Consider the game $\Gamma_{{\scriptscriptstyle \Bg}}^{n}$ on the Wheatstone network of Figure~\ref{fi:Wheatstone} in Example~\ref{ex:sequence-stoch-games}. 
Recall that there are $n$ players---all of the same type $t$---and each one is present with probability $1/n$. 
Assume that all players play a mixed strategy in which with probability $1/2$ they choose the upper path, and with probability $1/2$ they choose the lower path. 
Let $\boldsymbol{Y}=(Y_{1},Y_{2},Y_{3})$ denote the random vector that gives the number of players on the paths $s_{1}=(e_{1},e_4)$, 
$s_{2}=(e_{1},e_{3},e_5)$ and $s_{3}=(e_{2},e_5)$, respectively. Then $Y_{1},Y_{3}\sim\Binomial(n,(2n)^{-1})$ and 
$Y_{2}\equiv 0$, while the corresponding posterior distributions are $Y_{1}^{-t},Y_{3}^{-t} \sim\Binomial(n-1,(2n-2)^{-1})$
and $Y_{2}^{-t}\equiv 0$.
\end{example}

\subsection{Poisson Games} 

An important subclass of games with population uncertainty is obtained when the $N_{t}$'s are independent Poisson 
variables $N_{t}\sim\Poisson({d}_{t})$, with $d_{t}>0$, that is,
\begin{equation}
\label{eq:meas-Poisson}
\mu(\bar{\boldsymbol{n}})
=\Prob(N_{t}
=\bar{n}_{t},\forall t\in\mathscr{T})
=\prod_{t\in\mathscr{T}}\mathrm{e}^{-{d}_{t}} \frac{( {d}_{t})^{\bar{n}_{t}}}{\bar {n}_{t}!}.
\end{equation}
It is not difficult to see that such Poisson games are characterized by the fact that the posteriors $\mu(\argdot\mid t)$ coincide with $\mu$ for every $t\in\mathscr{T}$. 
Moreover, in this case the loads $Y_{t,s}$ are also independent with $Y_{t,s}\sim\Poisson({d}_{t}\sigma_{t}(s))$. 
In fact, as shown in \citet[Theorem 1]{Mye:IJGT1998}, in a game with population uncertainty 
the variables $Y_{t,s}$ are independent {\em if and only if}  the game is Poisson. 

It turns out that the nonatomic game $\breve{\Gamma}^{\infty}$ in Section~\ref{se:convergence-stochastic}, obtained as a limit of a sequence of Bernoulli congestion games, can 
be interpreted as a Poisson game defined by the costs
\begin{equation}
\label{eq:Poisson-cost}
C(\boldsymbol{Y}^{-t};t,s)
=\sum_{e\in s}c_{e}(1+X_{e}^{-t}), \quad X_{e}^{-t}=\sum_{t'\in\mathscr{T}}\sum_{s\in\mathscr{S}_{t'}}Y^{-t}_{t',s}\mathds{1}_{\braces{e\in s}}.
\end{equation}
We state this observation in the following result.

\begin{theorem} 
\label{th:Bernoulli-Poisson}
Let ${\boldsymbol{\sigma}}$ be a strategy profile in the Poisson game defined by the costs in \eqref{eq:Poisson-cost} and the demands given in \eqref{eq:meas-Poisson} with $d_{t}>0$ satisfying \eqref{eq:HP}.
Define $({\boldsymbol{y}},{\boldsymbol{x}})$ as ${y}_{t,s}={d}_{t}\sigma_{t}(s)$ and 
${x}_{e}=\sum_{t\in\mathscr{T}}\sum_{s\in\mathscr{S}_{t}}{y}_{t,s}\mathds{1}_{\braces{e\in s}}$.
Then ${\boldsymbol{\sigma}}$ is an equilibrium in the Poisson game if and only if $({\boldsymbol{y}},{\boldsymbol{x}})$ is a Wardrop equilibrium for the nonatomic game $\breve{\Gamma}^\infty$ in Section~\ref{se:convergence-stochastic}.
\end{theorem}

\proof{Proof.}
It suffices to note that the posterior distribution of $\boldsymbol{Y}^{-t}$ is Poisson with independent components and expected values
$\Expect_{\mu,\boldsymbol{\sigma}}[\boldsymbol{Y}^{-t}_{t',s}]={y}_{t',s}$, so that 
$X^{-t}_{e}\sim\Poisson({x}_{e})$. Then, taking expectation, we get precisely 
\begin{equation}
\label{eq:EPcost}
\Expect_{\mu,\boldsymbol{\sigma}}\bracks{C(\boldsymbol{Y}^{-t};t,s)}=\sum_{e\in s}\Expect\bracks{c_{e}(1+X^{-t}_{e})}=\sum_{e\in s}\breve c_{e}({x}_{e}),
\end{equation}
where $\breve{c}_{e}(\argdot)$ is defined as in \eqref{eq:cost-Poisson}.
\Halmos
\endproof

Note that the costs in \eqref{eq:Poisson-cost} depend on the strategy flows $Y^{-t}$ only through the aggregate 
resource loads $X_{e}^{-t}$, whereas \citet{Mye:IJGT1998} considers more general cost functions $C(Y^{-t};t,s)$. 
However, as already mentioned, Myerson required costs to be bounded to ensure that their expected values are well defined,
whereas for congestion games this boundedness can be relaxed and replaced by the conditions in \eqref{eq:HP}.

\subsection{Convergence of Bernoulli Games with Bounded Costs}

As mentioned in the introduction, \citet{Mye:IJGT1998} introduced Poisson games axiomatically. 
The following variation of Theorem~\ref{th:Bayes-Poisson} provides a justification for Poisson games 
as limits of a sequence of finite games with population uncertainty of Bernoulli type. 
The result goes beyond the separable cost structure \eqref{eq:Poisson-cost} of congestion games, 
allowing the cost of a player to depend on her type, her action, and the full vector of strategy flows. 
To compensate, the costs are required to be bounded.

Consider a finite set of types  $\mathscr{T}$ with strategy sets $\mathscr{S}_{t}$ for $t\in\mathscr{T}$, and 
 a sequence of games $\Gamma^{n}$ with finitely many players $i\in\mathscr{N}^{n}$, with types $t_{i}^{n}$, and probabilities $u_{i}^{n}$ of being active. 
Let $U_{i}^{n}$ be Bernoulli random variables with $\Prob(U_{i}^{n}=1)=u_{i}^{n}$,
and let player $i$ choose $S_{i}^{n}\in\mathscr{S}_{t_{i}^{n}}$ at random using a mixed strategy $\sigma_{i}^{n}$. 
Let $Y^{n}$ be the random vector of strategy flows $Y^{n}_{t,s}=\sum_{j\colon t_{j}^{n}=t} U_{j} \mathds{1}_{\braces{S_{j}^{n}=s}}$,
 and define $Y^{-i,n}$ similarly by excluding player $i$. The expected cost of an action $s\in\mathscr{S}_{t_{i}^{n}}$
 for player $i$ is given by 
$\Expect_{\widehat{\boldsymbol{\sigma}}^{n}}[C(Y^{-i,n};t_{i}^{n},s)]$. A Nash equilibrium $\widehat{\boldsymbol{\sigma}}^{n}$ is defined as usual by the condition
\begin{equation*}
\forall i\in\mathscr{N}^{n}, \forall s,s'\in\mathscr{S}_{t_{i}^{n}}\quad \widehat{\sigma}_{i}^{n}(s)>0
\implies \Expect_{\widehat{\boldsymbol{\sigma}}^{n}}[C(Y^{-i,n};t_{i}^{n},s)]
\leq \Expect_{\widehat{\boldsymbol{\sigma}}^{n}}[C(Y^{-i,n};t_{i}^{n},s')].
\end{equation*}

\begin{theorem}
\label{th:general-Poisson}
Consider a sequence of games $\Gamma^{n}$ as above, with bounded cost functions $C(\argdot;t,s)$, and assume the conditions in \eqref{eq:H3} with $d_{t}>0$.
Then, for every sequence $\widehat{\boldsymbol{\sigma}}^{n}$ of Nash equilibria, the expected loads $y_{t,s}^{n}=\Expect_{\widehat{\boldsymbol{\sigma}}^{n}}\bracks{Y_{t,s}^{n}}$ are bounded and each accumulation point $\widehat{\boldsymbol{y}}=(\widehat{y}_{t,s})_{t\in\mathscr{T},s\in\mathscr{S}_{t}}$ corresponds to an equilibrium $\widehat{\boldsymbol{\sigma}}$ in the 
Poisson game by setting $\widehat{\sigma}_{t}(s)=\widehat{y}_{t,s}/{d}_{t}$ for all $s\in\mathscr{S}_{t}$.
\end{theorem}

\proof{Proof.}
The boundedness of the expected strategy loads $y_{t,s}^{n}$ follows from $\sum_{s\in\mathscr{S}_{t}}y_{t,s}^{n}= {d}_{t}^{n}\to {d}_{t}$.
Take a convergent subsequence and rename it so that $\boldsymbol{y}^{n}\to\widehat{\boldsymbol{y}}$, and define $\widehat{\boldsymbol{\sigma}}$ as in the statement.
If $\widehat{\sigma}_{t}(s)>0$ for some $s\in\mathscr{S}_{t}$, then $\widehat{y}_{t,s}>0$, so that, for all $n$ large we have $y^{n}_{t,s}>0$ and there must be a player $i^{n}\in\mathscr{N}^{n}$ with type $t$ and $\widehat{\sigma}^{n}_{i^{n}}(s)>0$. 
The equilibrium condition for $i^{n}$ implies that 
\begin{equation}
\label{eq:peq2}
\forall s'\in \mathscr{S}_{t} \qquad \Expect_{\widehat{\boldsymbol{\sigma}}^{n}}\bracks{C(Y^{-i^n,n};t,s)}\leq \Expect_{\widehat{\boldsymbol{\sigma}}^{n}}\bracks{C(Y^{-i^n,n};t,s')}.
\end{equation}
As in Proposition~\ref{pr:Poisson-conv}, it follows that the variables $\boldsymbol{Y}^{-i^{n},n}$ converge in total variation, and hence in distribution, to a random 
vector $\boldsymbol{Y}$ with independent Poisson components $Y_{t,s}\sim\Poisson(\widehat{y}_{t,s})$.
Letting $n\to \infty$ in \eqref{eq:peq2} we obtain \eqref{eq:Poisson-eq}, from which the result follows.
\Halmos
\endproof

\begin{remark} Theorem~\ref{th:general-Poisson} remains valid when some of the demands converge to zero $d_{t}^{n}\to d_{t}=0$, 
by considering the Poisson game restricted to the nontrivial types with $d_{t}>0$.
\end{remark}

\begin{remark} Despite their similarity, Theorems~\ref{th:Bayes-Poisson} and \ref{th:general-Poisson} are independent and neither one follows from the other.
In fact, Theorem~\ref{th:general-Poisson} allows for more general forms of the costs as long as they are bounded, whereas Theorem~\ref{th:Bayes-Poisson} 
can handle unbounded costs provided that they have the specific additive structure in \eqref{eq:Poisson-cost} and under the conditions in \eqref{eq:HP}.
\end{remark}

\section{Approximation Bounds and Rates of Convergence}
\label{se:rate-of-convergence}

In this section we establish non-asymptotic  bounds for the distance between the distribution of the edge loads at 
equilibrium in the finite games and the loads in the corresponding Wardrop and Poisson games. These  
bounds provide quantitative estimates on how well a Wardrop or Poisson equilibrium approximates the equilibrium of the finite game. 
We highlight that these results require the costs to have derivatives bounded away from zero.

As in previous sections we analyze separately the case of weighted congestion games and Bernoulli congestion games.
In both settings we exploit estimates for the distance between Wardrop equilibria in non\-atomic games 
with different demands, as well as between exact and approximate equilibria. 
Recall that an $\varepsilon$-approximate Wardrop equilibrium $(\boldsymbol{y},\boldsymbol{x})\in{\mathscr{F}}(\boldsymbol{d})$, or $\varepsilon$-Wardrop equilibrium, is defined exactly like a Wardrop equilibrium up to an additive error:
\begin{equation}
\label{eq:w4eps}
\forall t\in\mathscr{T},\ \forall s,s'\in\mathscr{S}_{t}\quad {y}_{t,s}>0 \implies \sum_{e\in s}c_{e}({x}_{e})\leq \sum_{e\in s'}c_{e}({x}_{e})+\varepsilon.
\end{equation}

\begin{proposition}
\label{pr:load-loadeq} 
Let $\parens{\widehat{\boldsymbol{y}},\widehat{\boldsymbol{x}}}$ be a Wardrop equilibrium for a nonatomic congestion game $\Gamma^{\infty}$ with $d_{\tot}\leq d_{\max}$.
Suppose that there exists some constant $c_{\min}'>0$ such that 
\begin{equation*}
\forall e\in\mathscr{E}, \forall x\in[0,d_{\max}]\quad c_{\min}'\leq c_{e}'(x),
\end{equation*}
and let $\Xi=\sqrt{2 C/c_{\min}'}$ with $C\geq\sum_{e\in s}c_{e}(d_{\max})$ for all strategies $s\in\cup_{t\in\mathscr{T}}\mathscr{S}_{t}$.

\begin{enumerate}[\upshape(a)]
\item
\label{it:pr:load-loadeq-1}
If $\parens{\boldsymbol{y},\boldsymbol{x}}\in \mathscr{F}(\boldsymbol{d})$ is an $\varepsilon$-approximate Wardrop equilibrium, then $\norm{\boldsymbol{x}-\widehat{\boldsymbol{x}}}_{2}\leq \sqrt{\varepsilon d_{\max}/c_{\min}'}$.
\item
\label{it:pr:load-loadeq-2}
If $\parens{\widehat{\boldsymbol{y}}',\widehat{\boldsymbol{x}}'}$ is a Wardrop equilibrium for perturbed demands $\boldsymbol{d}'$ with $d_{\tot}'\leq d_{\max}$,
then $\norm{\widehat{\boldsymbol{x}}'-\widehat{\boldsymbol{x}}}_{2}\leq \Xi\cdot \sqrt{\norm{\boldsymbol{d}'-\boldsymbol{d}}_{1}}$.
\end{enumerate}
\end{proposition} 
\proof{Proof.} 
See Appendix~\ref{se:rate-of-convergence-proof}. 
\Halmos
\endproof

\subsection{Weighted Congestion Games}
\label{suse:rate-weighted}

Theorem~\ref{th:T1} shows that the random loads $X_{e}^{n}$ in the finite weighted congestion games converge in distribution
to the Wardrop equilibrium loads $\widehat{x}_{e}$. 
Under additional assumptions which call for lower and upper bounds on the derivatives of the cost functions, 
we can find nonasymptotic estimates for the rate of convergence. 

We first show that any mixed Nash equilibrium of a weighted congestion game $\Gamma_{{\scriptscriptstyle \Wg}}$ is close to a Wardrop equilibrium in a nonatomic game $\Gamma^\infty$ with the same aggregate demands. 
The distance depends---in a way that the theorem makes precise---on the topology of the instance, on the cost functions, and on the magnitude of the weights. 

\begin{theorem} 
\label{th:random-load-load-0}
Let $\Gamma_{{\scriptscriptstyle \Wg}}$ be a weighted congestion game with aggregate demands $(d_{t})_{t\in\mathscr{T}}$ given by \eqref{eq:demand-weights} with $d_{\tot}\leq d_{\max}$.
Assume that there exist constants
$c_{\max}'\geq c_{\min}'>0$ such that 
\begin{equation*}
\forall e\in\mathscr{E}, \forall x\in[0,d_{\max}],\quad c_{\min}'\leq c_{e}'(x) \leq c_{\max}'
\end{equation*}
and suppose in addition that the derivatives $c_{e}'(\argdot)$ are $L$-Lipschitz continuous for some $L\geq 0$.
Define $\theta=\sqrt{d_{\max}/4}+\sqrt{\kappa\,d_{\max} \parens*{2c_{\max}'+L\,d_{\max}/4}/c_{\min}'}$ with 
  $\kappa\geq |s|$ for all strategies $s\in\cup_{t\in\mathscr{T}}\mathscr{S}_{t}$.
Let $X_{e}$ be the random load in a mixed Nash equilibrium $\widehat{\boldsymbol{\sigma}}\in\MNE(\Gamma_{{\scriptscriptstyle \Wg}})$, and $\widehat{x}_{e}$ the 
unique resource load in the Wardrop equilibrium for the nonatomic game $\Gamma^\infty=\parens{\mathscr{G},\parens{d_{t}}_{t\in\mathscr{T}}}$
with the same aggregate demands as $\Gamma_{{\scriptscriptstyle \Wg}}$. 
Letting $w=\max_{i\in\mathscr{N}}w_{i}$, we have
\begin{equation}
\label{eq:L2-est-0}
\norm{X_{e}-\widehat{x}_{e}}_{L^{2}} \leq \theta\cdot\sqrt{w}.
\end{equation}
As a reminder, the $L^2$ norm is defined as $\norm{X}_{L^{2}}=\sqrt{\Expect\bracks{X^{2}}}$.
\end{theorem}

\proof{Proof.} 
See Appendix~\ref{se:rate-of-convergence-proof}. 
\Halmos
\endproof

An unfortunate feature of the estimate \eqref{eq:L2-est-0} is the presence of the lower bound $c_{\min}'$  in the coefficient $\theta$. In particular, for fixed demands the left hand side of \eqref{eq:L2-est-0}
remains bounded whereas the bound on the right diverges as $c_{\min}'\to 0$, so that this estimate is not always tight. 
However, for arbitrary demands, the following example shows that the discrepancy between the equilibrium loads in a weighted atomic game and in the Wardrop equilibrium can be arbitrarily large when $c_{\min}'$ is small, and both the right and left hand side of  \eqref{eq:L2-est-0} tend to infinity, although at different rates. 
Although this example was constructed to illustrate the dependence on $c_{\min}'$, the upper bound parameter $\theta$ depends on various other parameters. A full study of the rate of convergence in relation to each of them is a topic for future research.

\begin{example}
Consider the network in Figure~\ref{fig:GrafoTight} with demands $d_{1}, d_{2}$ from two different
origins and a common destination. Let $v$ and $z$ denote the fractions of the 
demands sent over the central path.

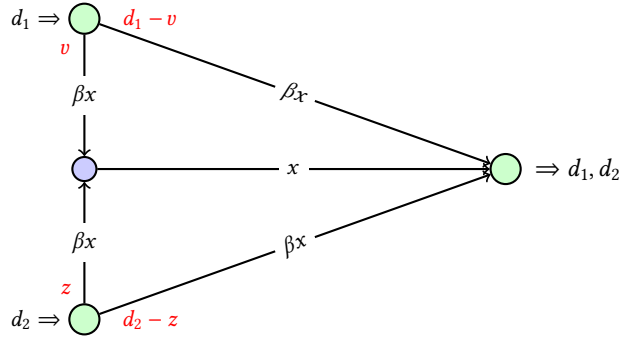
\begin{figure}[ht]
\centering
\begin{tikzpicture}[scale=0.8,node distance = 3 cm,thick,every node/.style={scale=0.8}]
  \tikzset{SourceNode/.style = {draw, circle, fill=green!20,minimum size=5mm}}
  \tikzset{DestNode/.style = {draw, circle, fill=green!20,minimum size=5mm}}
  \tikzset{NodeStyle/.style = {draw, circle, fill=blue!20,minimum size=4mm}}
     \node (D1) at (-1.3,2.5) {\small$d_{1}\Rightarrow$};
          \node (D1x) at (-0.8,2.0) {\small\color{red} $v$};
          \node (D1xx) at (0.6,2.5) {\small\color{red}$d_{1}-v$};
     \node (D2) at (-1.3,-2.5) {\small$d_{2}\Rightarrow$};
          \node (D2y) at (-0.8,-2.0) {\small\color{red} $z$};
          \node (D2yy) at (0.6,-2.5) {\small\color{red}$d_{2}-z$};
     \node (D12) at (7.7,0) {$\Rightarrow d_{1},d_{2}$};
     \node[SourceNode](s1) at (-0.5,2.5) {};
     \node[SourceNode](s2) at (-0.5,-2.5) {};
     \node[NodeStyle](v) at (-0.5,0) {};
     \node[DestNode](t) at (6.5,0) {};
     \draw[->](s1) to node[fill = white]{\small $\beta x$} (v);
     \draw[->](s2) to node[fill = white]{\small $\beta x$} (v);
     \draw[->](s1) to node[fill = white,sloped]{\small $\beta x$} (t);
     \draw[->](s2) to node[fill = white,sloped]{\small $\beta x$} (t);
     \draw[->](v) to node[fill = white,sloped]{\small $x$} (t);
\end{tikzpicture}
\caption{\label{fig:GrafoTight} A network with two OD pairs and unbounded errors for $\beta\to 0$. Labels over the arcs denote cost functions, whereas labels close to origins denote path flows.}
\end{figure}

\noindent Equalizing the costs of the central and outer paths we get the linear system
\begin{equation}
\label{eq:WE}
\begin{split}
(1+2\beta)\widehat{v} + \widehat{z}&=\beta d_{1}\\ 
\widehat{v}+(1+2\beta) \widehat{z} &=\beta d_{2},
\end{split}
\end{equation}
whose solution is a Wardrop equilibrium provided that $\widehat{v} \geq 0$ and $\widehat{z} \geq 0$. 
Similarly, for an atomic unsplittable game with weights $w_{i}\equiv w$ and demands $d_{1}=n_{1}w$, $d_{2}=n_{2}w$, integer multiples of $w$, 
we have that $v=k_{1}w$ and $z=k_{2}w$ with $k_{1},k_{2}\in\mathbb{N}\setminus\{0\}$ represent a pure Nash equilibrium if and only if 
\begin{equation}
\label{eq:NE}
\begin{split}
(1+2\beta) v + z &=\beta d_{1}+\varepsilon_{1}\\ 
 x + (1+2\beta) z &=\beta d_{2}+\varepsilon_{2}
\end{split}\quad\text{ with }
\varepsilon_{1},\varepsilon_{2} \in [-(1+\beta)w,\beta w].
\end{equation}
In particular, for $\varepsilon_{1}=0$ and $\varepsilon_{2}=-(1+\beta)w$ with ${1}/{\beta}\in\mathbb{N}$, it follows from \eqref{eq:NE} that the
demands  $d_{1}=w[2k_{1}+(k_{1}+k_{2})/\beta]$ and 
$d_{2}=w[2k_{2}+1+(k_{1}+k_{2}+1)/\beta]$ are integer multiples of $w$.
Now, solving \eqref{eq:WE} and \eqref{eq:NE} we get
\begin{equation*}
\begin{pmatrix} 
v\\ z 
\end{pmatrix}
-
\begin{pmatrix} 
\widehat{v}\\ \widehat{z} 
\end{pmatrix}
=\frac{w}{4\beta}
\begin{pmatrix} 
1\\ -(1+2\beta)
\end{pmatrix},
\end{equation*}
so that the nonnegativity conditions $\widehat{v}\geq 0$ and $\widehat{z}\geq 0$ are equivalent to $k_{1}\geq 1/(4\beta)$ and $k_{2}\geq -(1+2\beta)/(4\beta)$.
Hence, taking $k_{1}=\lceil 1/(4\beta)\rceil$ and $k_{2}=1$ we have that $(v,z)$ and $(\widehat{v},\widehat{z})$ are respectively Nash and Wardrop equilibria. 
Their load difference is of the order of $O(w/\beta)$, which diverges as $\beta\to 0$. 
Note that here $c_{\min}'=\beta$, $c_{\max}'=1$, 
$L=0$, $\kappa=2$, and
$d_{\max}=d_{1}+d_{2}\sim w/(2\beta^2)$, so that the bound \eqref{eq:L2-est-0}
becomes $\theta\sqrt{w}\sim\sqrt{2}w/\beta^{3/2}$,
which also diverges as $\beta\to 0$, although at a faster rate compared to the actual error $O(w/\beta)$. 
We do not know an instance when \eqref{eq:L2-est-0} is tight.
\end{example}

Combining Proposition~\ref{pr:load-loadeq} and Theorem~\ref{th:random-load-load-0} we can derive an explicit estimate for the $L^2$ distance between the random resource loads $X_{e}^{n}$ in a sequence of mixed Nash equilibria for $\Gamma_{{\scriptscriptstyle \Wg}}^{n}$, and the loads in the Wardrop equilibrium of the nonatomic limit game $\Gamma^\infty$.

The resulting bound in \eqref{eq:L2-est} has an additional term that involves the distance between the demands along the sequence $\Gamma_{{\scriptscriptstyle \Wg}}^{n}$
and those in $\Gamma^{\infty}$.
Thus, having an equilibrium close to the limit  requires small weights  and aggregate demands that are  close to the limiting demands.

\begin{corollary} 
\label{th:random-load-load}
Let $\Gamma_{{\scriptscriptstyle \Wg}}^{n}$ be a sequence of weighted congestion games satisfying \eqref{eq:H1} with $d_{\tot}^{n}\leq d_{\max}$ for all $n$.
Let the costs $c_{e}(\argdot)$ and $\theta$ be as in Theorem~\ref{th:random-load-load-0}, and  
$\Xi$ as in Proposition~\ref{pr:load-loadeq}. 
Let $X_{e}^{n}$ be the random loads in a sequence of mixed Nash equilibria $\widehat{\boldsymbol{\sigma}}^{n}\in\MNE(\Gamma_{{\scriptscriptstyle \Wg}}^{n})$, 
and $\widehat{x}_{e}$ the unique resource loads in the Wardrop equilibrium for the nonatomic limit game $\Gamma^\infty=\parens{\mathscr{G},\parens{d_{t}}_{t\in\mathscr{T}}}$. 
Then, with $w^{n}$ defined by \eqref{eq:w-n}, we have
\begin{equation}
\label{eq:L2-est}
\norm{X_{e}^{n}-\widehat{x}_{e}}_{L^{2}} \leq \theta\cdot\sqrt{w^{n}}+\Xi\cdot\sqrt{\norm{\boldsymbol{d}^{n}-\boldsymbol{d}}_{1}}.
\end{equation}
\end{corollary} 
\proof{Proof.} 
See Appendix~\ref{se:rate-of-convergence-proof}. 
\Halmos
\endproof

\subsection{Bernoulli Congestion Games}
\label{suse:rate-Bernoulli}

Analogously to Section~\ref{suse:rate-weighted}, we now find non-asymptotic estimates for the rate of convergence in Bernoulli congestion games.
Notice that \eqref{eq:rate} already provides a simpler---although difficult to quantify---bound on the distance between the resource loads in the 
finite games $\Gamma_{{\scriptscriptstyle \Bg}}^{n}$ and the Poisson limits. Indeed, although $ u^{n}$ is a primitive parameter of the model and its proximity 
to zero is readily available, the second term $\abs{\widehat{x}_{e}^{n}-\widehat{x}_{e}}$ is only known to converge to zero asymptotically. 
The following estimates provide more explicit bounds expressed directly in terms of the primitives of the model.

We first show that the distance between a Bayesian Nash equilibrium 
for a Bernoulli congestion game $\Gamma_{{\scriptscriptstyle \Bg}}$ and a Wardrop equilibrium for the nonatomic 
game $\breve{\Gamma}^\infty$ with the same aggregate demands
is bounded by an expression that depends on the topology and the function 
$\Lambda(u)$ defined in \eqref{eq:Lambda} evaluated at the maximum participation probability $u=\max_{i\in\mathscr{N}}u_{i}$.

\begin{theorem} 
\label{th:rate-Poisson-0}
Let $\Gamma_{{\scriptscriptstyle \Bg}}$ be a Bernoulli congestion game with expected demands $d_{t}=\sum_{i:t_{i}=t}u_{i}$ and $d_{\tot}\leq d_{\max}$.
Suppose that \eqref{eq:HP} holds and that
$\breve{c}_{e}'(x)\geq\breve{c}_{\min}'>0$ for all $x\in\bracks{0,d_{\max}}$. 
Let $\breve{\theta}=\sqrt{2\kappa d_{\max}/\breve{c}_{\min}'}$
with $\kappa\geq|s|$ for all strategies $s\in\cup_{t\in\mathscr{T}}\mathscr{S}_{t}$.
Let $X_{e}$ be the random loads in a Bayesian Nash equilibrium $\widehat{\boldsymbol{\sigma}}\in\BNE(\Gamma_{{\scriptscriptstyle \Bg}})$ with expected flow-loads $(\boldsymbol{y},\boldsymbol{x})$, 
and $\widehat{X}_{e}\sim\Poisson(\widehat{x}_{e})$ with $\widehat{\boldsymbol{x}}$ the 
unique Wardrop equilibrium loads in the nonatomic game $\breve{\Gamma}^{\infty}$ with costs $\breve{c}_{e}(\argdot)$ and
with the same aggregate demands $d_{t}$ as $\Gamma_{{\scriptscriptstyle \Bg}}$.
Then, 
$\norm{\boldsymbol{x}-\widehat{\boldsymbol{x}}}_{2}\leq \breve{\theta}\cdot\sqrt{\Lambda(u)}$
and 
\begin{equation}
\label{eq:B1}
\rho_{\TV}\parens*{\mathscr{L}(X_{e}),\mathscr{L}(\widehat{X}_{e})}\leq u + \breve{\theta}\cdot\sqrt{\Lambda(u)}.
\end{equation}
\end{theorem} 

\proof{Proof.} 
See Appendix~\ref{se:rate-of-convergence-proof}. 
\Halmos
\endproof

\begin{remark}
The condition $\breve{c}_{e}'(x)\geq\breve{c}_{\min}'$ holds when $c_{e}(k+1)\geq c_{e}(k)+\breve{c}_{\min}'$ for all $e\in\mathscr{E}$ and $k\geq 1$. 
A milder assumption is $0<\delta\coloneqq\min_{e\in\mathscr{E}}c_{e}(2)-c_{e}(1)$ in which case one can take $\breve{c}_{\min}'=\delta \mathrm{e}^{-d_{\max}}$.
\end{remark}

\begin{example} Similarly to the case of weighted congestion games, for $\breve{c}_{\min}'=\beta\approx 0$ the distance $\norm{\boldsymbol{x}-\widehat{\boldsymbol{x}}}_{2}$ between the loads in a Bernoulli congestion game 
and the corresponding Wardrop equilibrium can be arbitrarily large.
Indeed, consider the graph in Figure~\ref{fig:GrafoTight} and a Bernoulli game 
with identical participation probabilities $u_{i}\equiv u$ and expected demands
$d_{1}=n_{1}u$, $d_{2}=n_{2}u$, integer multiples of $u$. In this setting
$v=k_{1}u$ and $z=k_{2}u$, with $k_{1},k_{2}\in\mathbb{N}\setminus\{0\}$, is a pure Bayes-Nash equilibrium
 if and only if 
\begin{equation}
\label{eq:BNE}
\begin{split}
(1+2\beta) v +  z &=\beta d_{1}+\varepsilon_{1}\\ 
v + (1+2\beta) z &=\beta d_{2}+\varepsilon_{2}
\end{split}\qquad\text{ with }
\varepsilon_{1},\varepsilon_{2}\in[-1-\beta u,-1+(1+\beta)u].
\end{equation}
Taking $\varepsilon_{1}=-1+(1+\beta)u$ and $\varepsilon_{2}=-1-\beta u$, with $1/\beta\in\mathbb{N}$
and $1/(\beta u)\in\mathbb{N}$, the demands are indeed integer multiples of $u$, whereas
solving  \eqref{eq:WE} and \eqref{eq:BNE} we get
\begin{equation*}
\begin{pmatrix} v\\ z \end{pmatrix}
-\begin{pmatrix} \widehat{v}\\ \widehat{z} \end{pmatrix}
=\frac{u}{4\beta(1+\beta)}
\begin{pmatrix}1+2\beta(2+\beta-1/u)\\ 
-1-2\beta(1+\beta+1/u)
\end{pmatrix}.
\end{equation*}
Hence, by taking $k_{1}\geq (1+2\beta(2+\beta-1/u))/(4\beta(1+\beta))$ and
$k_{2}=1$  we have that $(\widehat{v},\widehat{z})$ is non-negative and therefore it is a Wardrop equilibrium, with
$\norm{\boldsymbol{x}-\widehat{\boldsymbol{x}}}_{2}\sim O(u/\beta)$ for $u$ fixed and $\beta\approx 0$.
The bound  in Theorem~\ref{th:rate-Poisson-0} is again not tight because 
$\breve{\theta}\sqrt{\Lambda(u)}\sim O(u/\beta^{3/2})$.
\end{example}

Using Theorem~\ref{th:rate-Poisson-0} we may estimate the distance between the random loads in a 
sequence $\Gamma_{{\scriptscriptstyle \Bg}}^{n}$ of Bernoulli congestion games and the corresponding Wardrop equilibrium in the limit game $\breve{\Gamma}^{\infty}$.
The bound in \eqref{eq:rate2}  has an additional term that involves the distance between the demands along the sequence $\Gamma_{{\scriptscriptstyle \Bg}}^{n}$
and the demands in $\breve{\Gamma}^{\infty}$.
Hence, having an equilibrium close to the limit  requires small participation probabilities together with the aggregate demands being close to the limiting demands.

\begin{corollary}
\label{co:rate-Poisson}
Let $\Gamma_{{\scriptscriptstyle \Bg}}^n$ be a sequence of Bernoulli congestion games satisfying \eqref{eq:H3} and \eqref{eq:HP} with
${d}_{\tot}^{n}\le d_{\max}$. 
Suppose that $\breve{c}_{e}'(x)\geq\breve{c}_{\min}'>0$ for all $x\in\bracks{0,d_{\max}}$, and
let $\breve{\theta}=\sqrt{2\kappa d_{\max}/\breve{c}_{\min}'}$  and $\breve{\Xi}=\sqrt{2{C}/\breve{c}_{\min}'}$
where $\kappa\geq |s|$ and $C\geq \sum_{e\in s}\breve{c}_{e}(d_{\max})$
for all $s\in\cup_{t\in\mathscr{T}}\mathscr{S}_{t}$. 
Let $X_{e}^{n}$ be the random load in a sequence of Bayes-Nash equilibria $\widehat{\boldsymbol{\sigma}}^{n}\in\BNE(\Gamma_{{\scriptscriptstyle \Bg}}^{n})$, 
and $X_{e}\sim\Poisson(\widehat{x}_{e})$ with $\widehat{x}_{e}$ the 
unique resource load in the Wardrop equilibrium for the nonatomic limit game $\breve{\Gamma}^{\infty}$. Then,
\begin{equation}
\label{eq:rate2}
\rho_{\TV}\parens*{\mathscr{L}(X_{e}^{n}),\mathscr{L}(X_{e})}\leq u^{n}
+\breve{\theta}\cdot\sqrt{\Lambda(u^{n})}+
\breve{\Xi}\cdot\sqrt{\norm{{\boldsymbol{d}}^{n}-{\boldsymbol{d}}}_{1}}.
\end{equation}
\end{corollary}

\proof{Proof.} 
See Appendix~\ref{se:rate-of-convergence-proof}. 
\Halmos
\endproof

\section{Convergence of the Inefficiency of Equilibria}
\label{se:PoA}

In this section we study the convergence of the inefficiency of equilibria, as captured by the PoA and the PoS.
Let us begin by recalling these notions.
We measure the social cost of a strategy profile as the sum of all players' costs. 
This provides us with a yardstick with which we can quantify the efficiency of equilibria as first proposed by \citet{KouPap:STACS1999}. 
The \emph{price of anarchy}  is the worst-case ratio of the social cost of the equilibrium to the optimum. 
Here, the worst case is taken with respect to all possible equilibria. 
The \emph{price of stability}  is defined accordingly with respect to the best equilibria.

Starting with nonatomic congestion games, their \emph{social cost} is given by
\begin{equation}
\label{eq:SC}
\forall\parens{\boldsymbol{y},\boldsymbol{x}}\in\mathscr{F}(\boldsymbol{d})\quad
\SC(\boldsymbol{y},\boldsymbol{x})\coloneqq\sum_{e\in\mathscr{E}}x_{e}\,c_{e}(x_{e}),
\end{equation}
from where the social optimum is 
$\Opt(\Gamma^{\infty})\coloneqq\min_{\parens{\boldsymbol{y},\boldsymbol{x}}\in\mathscr{F}(\boldsymbol{d})}\SC(\boldsymbol{y},\boldsymbol{x})$. 
It is well known \citep[see][]{BecMcGWin:Yale1956} that, whenever the cost functions $c_{e}$ are weakly increasing, all Wardrop equilibria have the same social cost, so that defining $\Eq(\Gamma^{\infty})\coloneqq\SC(\widehat{\boldsymbol{y}},\widehat{\boldsymbol{x}})$
for any $\parens{\widehat{\boldsymbol{y}},\widehat{\boldsymbol{x}}}\in\WE(\Gamma^{\infty})$, it follows that the PoA and PoS for a nonatomic game coincide and 
are given by
\begin{equation*}
\PoA(\Gamma^{\infty})=\PoS(\Gamma^{\infty})\coloneqq\frac{\Eq(\Gamma^{\infty})}{\Opt(\Gamma^{\infty})}.
\end{equation*}

The corresponding definitions for weighted congestion games and Bernoulli congestion games are similar, adjusted for the fact that now these games include stochastic realizations. 
The \emph{expected social cost} is
\begin{equation}
\label{eq:ESC}
\forall\boldsymbol{\sigma}\in\Sigma\quad
\ESC(\boldsymbol{\sigma})\coloneqq\Expect_{\boldsymbol{\sigma}}\bracks*{\sum_{e\in\mathscr{E}}X_{e}\,c_{e}(X_{e})},
\end{equation}
where $\parens{X_{e}}_{e\in\mathscr{E}}$ are the random resource loads induced by the mixed strategy profile $\boldsymbol{\sigma}$.
The optimum cost is
$\Opt(\Gamma)\coloneqq\min_{\boldsymbol{\sigma}\in\Sigma}\ESC(\boldsymbol{\sigma})$, 
and the social optimum is 
$\widetilde{\boldsymbol{\sigma}}\in\argmin_{\boldsymbol{\sigma}\in\Sigma}\ESC(\boldsymbol{\sigma})$.
Considering the worst and best social cost at equilibrium,
the price of anarchy and stability are:
\begin{equation*}
\PoA(\Gamma)\coloneqq\max_{\widehat{\boldsymbol{\sigma}}\in\MNE(\Gamma)}\frac{\ESC(\widehat{\boldsymbol{\sigma}})}{\Opt(\Gamma)}, \text{\quad and \quad}
\PoS(\Gamma)\coloneqq\min_{\widehat{\boldsymbol{\sigma}}\in\MNE(\Gamma)}\frac{\ESC(\widehat{\boldsymbol{\sigma}})}{\Opt(\Gamma)}.
\end{equation*}

With these definitions we may now establish the convergence of the PoA for a sequence of weighted congestion games 
$\Gamma_{{\scriptscriptstyle \Wg}}^{n}$. The proof uses some auxiliary results on the convergence of social costs presented in Appendix~\ref{se:PoA-Proofs}.

\begin{theorem}
\label{th:ESC-to-Eq}
Let $\Gamma_{{\scriptscriptstyle \Wg}}^{n}$ be a sequence of weighted congestion games satisfying the conditions in \eqref{eq:H1} and let $\widehat{\boldsymbol{\sigma}}^{n}\in\MNE(\Gamma_{{\scriptscriptstyle \Wg}}^{n})$. 
Then, 
$\ESC(\widehat{\boldsymbol{\sigma}}^{n})\to\Eq(\Gamma^{\infty})$, and 
therefore both the price of anarchy $\PoA(\Gamma_{{\scriptscriptstyle \Wg}}^{n})$ and the price of stability $\PoS(\Gamma_{{\scriptscriptstyle \Wg}}^{n})$
converge towards $\PoA(\Gamma^{\infty})=\PoS(\Gamma^{\infty})$.
\end{theorem}

\proof{Proof.}
From Theorem~\ref{th:T1} we know that every accumulation point of $(\widehat{\boldsymbol{y}}^{n},\widehat{\boldsymbol{x}}^{n})$ is a Wardrop equilibrium for $\Gamma^\infty$; hence, by Lemma~\ref{le:convergence-SC}, the full sequence $\ESC(\widehat{\boldsymbol{\sigma}}^{n})$ converges to $\Eq(\Gamma^\infty)$ as $n\to\infty$. On the other hand, Proposition~\ref{pr:min-SC} shows that $\Opt(\Gamma_{{\scriptscriptstyle \Wg}}^{n})\to\Opt(\Gamma^{\infty})$ from which the convergence
of $\PoA(\Gamma_{{\scriptscriptstyle \Wg}}^{n})$ and  $\PoS(\Gamma_{{\scriptscriptstyle \Wg}}^{n})$ follows at once.
\Halmos
\endproof

\begin{example}
\label{ex:sequence-weighted-PoA}
Consider the sequence of games in Example~\ref{ex:Weathstone-weighted} and the different equilibria described there. The social cost is minimized by splitting
half of the players between the upper and lower paths (up to 1 player when $n$ is odd). For $n=2$, we have $\PoA(\Gamma_{n})=4/3$ and 
$\PoS(\Gamma_{n})=1$, whereas for $n\geq3$, by setting $\delta_n=1$ if $n$ is odd and $\delta_n=0$ otherwise, we get $\PoA(\Gamma^{n})=4n^2/(3n^2+\delta_n)$ and $\PoS(\Gamma^{n})=(4n^2-2n+2)/(3n^2+\delta_n)$,  both converging to $\PoA(\Gamma^{\infty})=\PoS(\Gamma^{\infty})=4/3$ as $n\to\infty$.
\end{example}

The following analogous result holds for sequences of Bernoulli congestion games.  

\begin{theorem}
\label{th:PoA}
Let $\Gamma_{{\scriptscriptstyle \Bg}}^{n}$ be a sequence of Bernoulli congestion games satisfying the conditions in \eqref{eq:H3} and \eqref{eq:HP}. 
Then, for every sequence $\widehat{\boldsymbol{\sigma}}^{n}\in\BNE(\Gamma_{{\scriptscriptstyle \Bg}}^{n})$, the expected social cost $\ESC(\widehat{\boldsymbol{\sigma}}^{n})$ converges to $\Eq(\breve{\Gamma}^{\infty})$.
As a consequence, both $\PoA(\Gamma_{{\scriptscriptstyle \Bg}}^{n})$ and $\PoS(\Gamma_{{\scriptscriptstyle \Bg}}^{n})$ converge to $\PoA(\breve{\Gamma}^{\infty})=\PoS(\breve{\Gamma}^{\infty})$.
\end{theorem}

\proof{Proof.}
From Lemma~\ref{le:aux} we have that $\breve{c}_{e}'(x)\geq 0$ for $x\in[0,d_{\max}]$ so that the extended costs $\breve{c}_{e}'(\argdot)$ 
are weakly increasing, and therefore the social cost $\SC(\widehat{\boldsymbol{y}},\widehat{\boldsymbol{x}})\equiv \Eq(\breve{\Gamma}^{\infty})$ is the same in every Wardrop equilibrium. 
Now, Theorem~\ref{th:Bayes-Poisson} states that every accumulation point of $(\widehat{\boldsymbol{y}}^{n},\widehat{\boldsymbol{x}}^{n})$ is a Wardrop equilibrium for $\breve{\Gamma}^\infty$; hence, by Lemma~\ref{le:convergence-SC-Poisson}, the full sequence $\ESC(\widehat{\boldsymbol{\sigma}}^{n})$ converges to $\Eq(\breve{\Gamma}^\infty)$ as $n\to\infty$. On the other hand, Proposition~\ref{pr:min-SC-Poisson} gives $\Opt(\Gamma_{{\scriptscriptstyle \Bg}}^{n})\to\Opt(\breve{\Gamma}^{\infty})$ from which the convergence
of $\PoA(\Gamma_{{\scriptscriptstyle \Bg}}^{n})$ and  $\PoS(\Gamma_{{\scriptscriptstyle \Bg}}^{n})$ follows at once.
\Halmos
\endproof

\begin{example}
\label{ex:sequence-stoch-PoA}
Consider the sequence of games $\Gamma_{{\scriptscriptstyle \Bg}}^{n}$ on the Wheatstone network in Example~\ref{ex:sequence-stoch-games}. 
The social cost is minimized by splitting
half of the players between the upper and lower routes (up to one player for $n$ odd).
This strategy profile is also a pure Nash equilibrium so that $\PoS(\Gamma_{{\scriptscriptstyle \Bg}}^{n})=1$ for all $n\in\mathbb{N}$.
As far as the PoA is concerned, the worst equilibrium occurs when each player chooses the upper and lower routes
with probability $1/2$. Setting $\delta_{n}=1/n$ when $n$ is odd and $\delta_{n}=0$ otherwise, we obtain $\PoA(\Gamma_{{\scriptscriptstyle \Bg}}^{n})=(5n-1)/(5n-2+\delta_{n})$,
which converges to $\PoA(\breve{\Gamma}^{\infty})=1$ as $n\to\infty$.
\end{example}

\begin{remark}
For polynomial costs $c_{e}(\argdot)$ of degree at most $d$ we have that $\breve c_{e}(x)$ are again 
polynomials of the same degree (though with different coefficients), so that the results in \citet{Rou:JCSS2003} imply 
\begin{equation*}
\lim_{n\to\infty}\PoA(\Gamma_{{\scriptscriptstyle \Bg}}^{n})=\PoA(\breve{\Gamma}^{\infty})\leq B(d)\coloneqq\frac{(d+1)\sqrt[d]{d+1}}{(d+1)\sqrt[d]{d+1}-d}.
\end{equation*}
In fact, for $d=1$, the bound $\PoA(\Gamma_{{\scriptscriptstyle \Bg}}^{n})\leq 4/3$ is valid as soon as $u^{n}\leq 1/4$ \citep[see][]{ComScaSchSti:arXiv2022}.
For higher degrees we conjecture the existence of a threshold for $u^{n}$ under which $\PoA(\Gamma_{{\scriptscriptstyle \Bg}}^{n})$ already falls below the nonatomic bound $B(d)$. 
The current result only implies that we  have this as an asymptotic bound when $u^{n}\to 0$.
\end{remark}

\begin{remark}
Examples~\ref{ex:sequence-weighted-PoA} and \ref{ex:sequence-stoch-PoA} may suggest that the PoA in the nonatomic game obtained as a limit of weighted congestion games would be larger 
than the PoA in the nonatomic limit game for Bernoulli congestion games. 
This is not true in general. 
Consider for instance the Pigou network in Figure~\ref{fi:Pigou} with a demand of $1$.
\begin{figure}[ht]
\centering
\begin{tikzpicture}[node distance = 3 cm,thick,every node/.style={scale=0.8}]
  \tikzset{SourceNode/.style = {draw, circle, fill=green!20}}
  \tikzset{DestNode/.style = {draw, circle, fill=green!20}}
     \node[SourceNode](A){$\mathsf{O}$};
     \node[DestNode,right=of A](B){$\mathsf{D}$};
     \draw[->, bend left = 45](A) to node[fill=white]{$c_{1}(x)=x$} (B);
     \draw[->,bend right = 45 ](A) to node[fill=white]{$c_{2}(x)=2$} (B);
\end{tikzpicture}
\caption{\label{fi:Pigou} Pigou network.}
\end{figure}
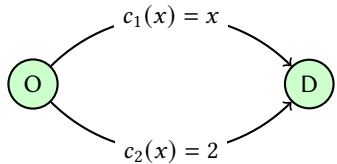
The Wardrop equilibrium of the standard nonatomic game is optimal, so $\PoA(\Gamma^{\infty})=1$. The Wardrop equilibrium of the nonatomic limit game of the Bernoulli game, in which the auxiliary cost function on the upper edge is now $\breve{c}_{1}(x)=1+x$, sends all demand on the upper path, whereas in the social optimum  the demand is split over the upper and lower path. So, we have that $\PoA(\breve{\Gamma}^{\infty})=8/7$.
\end{remark}

\section{Conclusions and Perspectives}
\label{se:conclusions}

In this paper we studied the convergence of equilibria of atomic unsplittable congestion games with an increasing number of players, towards a Wardrop equilibrium for a limiting nonatomic game. 
For the case where players have vanishing weights, the random flow-load pairs $(\boldsymbol{Y}^{n},\boldsymbol{X}^{n})$ at a mixed equilibrium in the finite games were shown to converge in 
distribution towards constant flow-load pairs $(\widehat{\boldsymbol{y}},\widehat{\boldsymbol{x}})$, which are Wardrop equilibria for the nonatomic limit game. 
In contrast, if players have a fixed unit weight but are present in the game with vanishing probabilities, then $(\boldsymbol{Y}^{n},\boldsymbol{X}^{n})$ converge in total variation to Poisson variables $(\boldsymbol{Y},\boldsymbol{X})$, whose expected values $(\widehat{\boldsymbol{y}},\widehat{\boldsymbol{x}})$ are again characterized as a Wardrop equilibrium for a nonatomic congestion game with auxiliary cost functions. 
In this case, the limit variables $(\boldsymbol{Y},\boldsymbol{X})$ can also be interpreted as an equilibrium for a Poisson game in the sense of \citet{Mye:IJGT1998}.
Under additional 
conditions we also established explicit estimates for the distance between the random loads $(\boldsymbol{Y}^{n},\boldsymbol{X}^{n})$
and their limits $(\widehat{\boldsymbol{y}},\widehat{\boldsymbol{x}})$ and $(\boldsymbol{Y},\boldsymbol{X})$, respectively.
These convergence results were completed by showing that in both frameworks the PoA and the PoS converge to the PoA of the limit game. 

We did not address the combined case in which both the weights and the presence probabilities vary across players. 
Such situations may be relevant for routing games where cars and trucks have a 
different impact on traffic, or in the presence of heterogeneous drivers that may be slower or faster inducing more or less congestion. 
Other settings in which players are naturally heterogeneous arise in telecommunications, where packets come in different sizes, and in processor sharing, where tasks arriving to a server have different workloads. 
In such cases, one might still expect to obtain a limit game which is likely to yield a \emph{weighted nonatomic game}, possibly involving weighted sums of Poisson distributions.

Another direction not explored in this paper concerns the case of oligopolistic competition in which some players, e.g., TomTom, Waze, FedEx, UPS, etc.\ may control a nonnegligible fraction of the demand, while simultaneously another fraction of the demand behaves as individual selfish players.
A natural conjecture is that in the vanishing weight limit one may converge
to a composite game as those studied by \citet{ComCorSti:OR2009} and \citet{SorWan:JDG2016}, with coexistence of atomic-splittable and nonatomic players. 
Similarly, in the case of fixed weights and vanishing probabilities one may expect to converge to some form of composite game involving Poisson random variables, which remains to be discovered.
%

%
%
%

\begin{APPENDICES}

\section{Approximation Bounds and Rates of Convergence -- Proofs} 
\label{se:rate-of-convergence-proof}

\proof{Proof of Proposition~\ref{pr:load-loadeq}.}
\ref{it:pr:load-loadeq-1}
Defining $C_{s}\coloneqq\sum_{e\in s} c_{e}(x_{e})$ and $\eta_{t}\coloneqq\min_{s\in\mathscr{S}_{t}}C_{s}$, we bound the squared distance by
\begin{align*}
\sum_{e\in\mathscr{E}}c_{\min}'\,(x_{e}-\widehat{x}_{e})^{2}&\leq\sum_{e\in\mathscr{E}}( c_{e}(x_{e})- c_{e}(\widehat{x}_{e}))(x_{e}-\widehat{x}_{e})\\
&\leq\sum_{e\in\mathscr{E}} c_{e}(x_{e})(x_{e}-\widehat{x}_{e})\\[-1.5ex]
&=\sum_{t\in\mathscr{T}}\sum_{s\in\mathscr{S}_{t}}\parens*{\sum_{e\in s} c_{e}(x_{e})}(y_{t,s}-\widehat{y}_{t,s})\\
&=\sum_{t\in\mathscr{T}}\sum_{s\in\mathscr{S}_{t}}(C_{s}-\eta_{t})(y_{t,s}-\widehat{y}_{t,s})\\
&\leq\sum_{t\in\mathscr{T}}\sum_{s\in\mathscr{S}_{t}}(C_{s}-\eta_{t})\,y_{t,s}\,.
\end{align*}
Here, in the second inequality we dropped the sum $\sum_{e\in\mathscr{E}} c_{e}(\widehat{x}_{e})(x_{e}-\widehat{x}_{e})$,
which is nonnegative because $(\widehat{\boldsymbol{y}},\widehat{\boldsymbol{x}})$ is a Wardrop equilibrium for $\Gamma^{\infty}$ and $(\boldsymbol{y},\boldsymbol{x})\in\mathscr{F}(\boldsymbol{d})$. {The first equality follows from expressing the resource loads in terms of the strategy flows and exchanging the order of summation, and the second equality from the fact that $\sum_{s\in\mathscr{S}_{t}}(y_{t,s}-\widehat{y}_{t,s})=0$ for all types $t$ as
 both solutions $\parens{\widehat{\boldsymbol{y}},\widehat{\boldsymbol{x}}}$ and $\parens{\boldsymbol{y},\boldsymbol{x}}$ are feasible.}
Now, using \eqref{eq:w4eps}, we conclude
\begin{equation*}
c_{\min}'\norm{\boldsymbol{x}-\widehat{\boldsymbol{x}}}_{2}^{2}\leq\sum_{t\in\mathscr{T}}\sum_{s\in\mathscr{S}_{t}}\varepsilon\,y_{t,s}=\varepsilon\sum_{t\in\mathscr{T}}d_{t}=\varepsilon\,d_{\tot}\leq\varepsilon\,d_{\max}. 
\end{equation*}

\noindent \ref{it:pr:load-loadeq-2}
Proceeding as in the previous part, we have
\begin{align*}
c_{\min}'\norm{\widehat{\boldsymbol{x}}'-\widehat{\boldsymbol{x}}}_{2}^{2}&\leq\sum_{e\in\mathscr{E}}(c_{e}(\widehat{x}_{e}')-c_{e}(\widehat{x}_{e}))(\widehat{x}_{e}'-\widehat{x}_{e})\\
&=\sum_{e\in\mathscr{E}} c_{e}(\widehat{x}_{e}')(\widehat{x}_{e}'-\widehat{x}_{e})
+\sum_{e\in\mathscr{E}} c_{e}(\widehat{x}_{e})(\widehat{x}_{e}-\widehat{x}_{e}').
\end{align*}
Let the two sums in the RHS be denoted by $\Psi(\widehat{\boldsymbol{x}}',\widehat{\boldsymbol{x}})$ and $\Psi(\widehat{\boldsymbol{x}},\widehat{\boldsymbol{x}}')$, respectively. 
To bound $\Psi(\widehat{\boldsymbol{x}}',\widehat{\boldsymbol{x}})$, we exploit the fact that $\widehat{\boldsymbol{x}}'$ is a Wardrop equilibrium for the demand $\boldsymbol{d}'$. 
Because $\parens{\widehat{\boldsymbol{y}},\widehat{\boldsymbol{x}}}$ is feasible for $\boldsymbol{d}$, but not for $\boldsymbol{d}'$, for each type $t$ with 
$d_{t}>0$ we consider the rescaled flows $y_{t,s}=\widehat{y}_{t,s} d_{t}'/d_{t}$, 
whereas when $d_{t}=0$ we simply take $y_{t,s}=\widehat{y}_{t,s}'$. 
Letting 
$x_{e}=\sum_{t\in\mathscr{T}}\sum_{s\in\mathscr{S}_{t}} y_{t,s}\mathds{1}_{\braces{e\in s}}$ denote the corresponding resource loads, 
we have that $(\boldsymbol{y},\boldsymbol{x})\in\mathscr{F}(\boldsymbol{d}')$, and therefore 
\begin{align*}
\Psi(\widehat{\boldsymbol{x}}',\widehat{\boldsymbol{x}})
&=\sum_{e\in\mathscr{E}} c_{e}(\widehat{x}_{e}')(\widehat{x}_{e}'-x_{e})+\sum_{e\in\mathscr{E}} c_{e}(\widehat{x}_{e}')(x_{e}-\widehat{x}_{e})\\
&\leq\sum_{e\in\mathscr{E}} c_{e}(\widehat{x}_{e}')(x_{e}-\widehat{x}_{e})\\[-1.5ex]
&=\sum_{t\in\mathscr{T}}\sum_{s\in\mathscr{S}_{t}}\parens*{\sum_{e\in s} c_{e}(\widehat{x}_{e}')}(y_{t,s}-\widehat{y}_{t,s}),
\end{align*}
where in the inequality we dropped the first sum which is nonpositive because $(\widehat{\boldsymbol{y}}',\widehat{\boldsymbol{x}}')$ is a Wardrop equilibrium and $(\boldsymbol{y},\boldsymbol{x})$ is feasible for $\boldsymbol{d}'$, 
whereas the last equality follows by expressing the resource loads in terms of the strategy flows and exchanging the order of the sums.

We now analyze each term in the outer sum over $t\in\mathscr{T}$. 
When $d_{t}>0$ the inner double sum can be bounded as 
\begin{align*}
\sum_{s\in\mathscr{S}_{t}}\parens*{\sum_{e\in s} c_{e}(\widehat{x}_{e}')}(y_{t,s}-\widehat{y}_{t,s})
&=\sum_{s\in\mathscr{S}_{t}}\parens*{\sum_{e\in s} c_{e}(\widehat{x}_{e}')}\widehat{y}_{t,s}(d_{t}'/d_{t}-1)\\
&\leq \sum_{s\in\mathscr{S}_{t}}C\,\widehat{y}_{t,s}\abs{d_{t}'/d_{t}-1}\\
&=C\abs{d_{t}'-d_{t}},
\end{align*}
whereas, when $d_{t}=0$, we have $y_{t,s}=\widehat{y}_{t,s}'$ and $\widehat{y}_{t,s}=0$ so that
\begin{equation*}
\sum_{s\in\mathscr{S}_{t}}\parens*{\sum_{e\in s} c_{e}(\widehat{x}_{e}')}(y_{t,s}-\widehat{y}_{t,s})
=\sum_{s\in\mathscr{S}_{t}}\parens*{\sum_{e\in s} c_{e}(\widehat{x}_{e}')}\widehat{y}_{t,s}'
\leq Cd_{t}'\\
=C\abs{d_{t}'-d_{t}}.
\end{equation*}
Summing these estimates over all $t\in\mathscr{T}$ we get $\Psi(\widehat{\boldsymbol{x}}',\widehat{\boldsymbol{x}})\leq C\norm{\boldsymbol{d}'-\boldsymbol{d}}_{1}$.
Symmetrically, we have $\Psi(\widehat{\boldsymbol{x}},\widehat{\boldsymbol{x}}')\leq C\norm{\boldsymbol{d}-\boldsymbol{d}'}_{1}$, from which the result follows.
\Halmos
\endproof

\proof{Proof of Theorem~\ref{th:random-load-load-0}.}
Let $x_{e}=\Expect_{\widehat{\boldsymbol{\sigma}}}\bracks{X_{e}}$ so that 
\begin{equation}
\label{eq:triangular}
\norm{X_{e}-\widehat{x}_{e}}_{L^{2}} \leq \norm{X_{e}-x_{e}}_{L^{2}} 
+\abs{x_{e}-\widehat{x}_{e}},
\end{equation}
with
\begin{equation}
\label{eq:variance}
\norm{X_{e}-x_{e}}_{L^{2}}^{2}=\Var_{\widehat{\boldsymbol{\sigma}}}(X_{e})
=\sum_{i\in\mathscr{N}}w_{i}^{2}\sigma_{i,e}(1-\sigma_{i,e})
\leq\frac{1}{4}\sum_{i\in \mathscr{N}}w_{i}^{2}
\leq\frac{w\,d_{\tot}}{4}
\leq\frac{w\,d_{\max}}{4}.
\end{equation}
With this bound in place, we proceed to estimate $\abs{x_{e}-\widehat{x}_{e}}$ by showing that 
$(\boldsymbol{y},\boldsymbol{x})$ is an $\varepsilon$-Wardrop equilibrium with 
\begin{equation*}
\varepsilon=\kappa\parens*{2c_{\max}'+ L\,d_{\max}/4}w.
\end{equation*}
To this end, we first observe that since the derivatives $c_{e}'(\argdot)$  are $L$-Lipschitz, 
we have:
\begin{equation*}
-\frac{1}{2}L\,(X_{e}-x_{e})^{2}
\leq c_{e}\big(X_{e}\big)-c_{e}(x_{e})-c_{e}'(x_{e})(X_{e}-x_{e})
\leq\frac{1}{2}\, L\,(X_{e}-x_{e})^{2},
\end{equation*}
so that taking expectations, and noting that the expected value of the linear part vanishes, we get
\begin{equation}
\label{eq-gamma-Var}
\abs*{\Expect_{\widehat{\boldsymbol{\sigma}}}\bracks*{c_{e}(X_{e})}-c_{e}(x_{e})}
\leq
\frac{1}{2}\, L\Var_{\widehat{\boldsymbol{\sigma}}}(X_{e})\leq L\,w\,d_{\max}/8.
\end{equation}
Now, if $y_{t,s}>0$, then there exists some player $i\in\mathscr{N}$ with $t_{i}=t$ and $\widehat{\sigma}_{i}(s)>0$, so that 
the equilibrium condition in $\Gamma_{{\scriptscriptstyle \Wg}}$ implies that for each alternative strategy $s'\in\mathscr{S}_{t}$ we have 
\begin{equation}
\label{eq:disc-0}
\sum_{e\in s} \Expect_{\widehat{\boldsymbol{\sigma}}}\bracks*{c_{e}(X_{i,e})}\leq \sum_{e\in s'}\Expect_{\widehat{\boldsymbol{\sigma}}}\bracks*{c_{e}(X_{i,e})}.
\end{equation}
Now, because $\abs{X_{i,e}-X_{e}} \leq w$ and $c_{e}(\argdot)$ is $c_{\max}'$-Lipschitz, by using \eqref{eq-gamma-Var} we get
\begin{equation*}
\abs*{\Expect_{\widehat{\boldsymbol{\sigma}}}\bracks*{c_{e}(X_{i,e})}-c_{e}(x_{e})}
\leq c_{\max}'w+\abs*{\Expect_{\widehat{\boldsymbol{\sigma}}}\bracks*{c_{e}(X_{e})}-c_{e}(x_{e})}
\leq \parens*{c_{\max}' +L\,d_{\max}/8}w
\end{equation*}
and, then, we can approximate \eqref{eq:disc-0} with respect to costs:
\begin{equation*}
\sum_{e\in s} c_{e}(x_{e})\leq \sum_{e\in s'} c_{e}(x_{e})+2\kappa\parens*{c_{\max}' + L d_{\max}/8}w.
\end{equation*}
This shows that $(\boldsymbol{y},\boldsymbol{x})$ is an $\varepsilon$-Wardrop equilibrium for the nonatomic game with demands $\boldsymbol{d}$, and then invoking Proposition~\ref{pr:load-loadeq}\ref{it:pr:load-loadeq-1} we get
\begin{equation}
\label{eq:combined-ineq}
\abs{x_{e}-\widehat{x}_{e}}
\leq\norm{\boldsymbol{x}-\widehat{\boldsymbol{x}}}_{2}\leq \sqrt{ \kappa d_{\max}\parens*{2c_{\max}' +L\,d_{\max}/4}w/c_{\min}'}.
\end{equation}
Plugging \eqref{eq:combined-ineq} and \eqref{eq:variance} into \eqref{eq:triangular} we obtain the final estimate in \eqref{eq:L2-est-0}.
\Halmos
\endproof

\proof{Proof of Theorem~\ref{th:random-load-load}.}
Take $(\widehat{\boldsymbol{y}}^{n},\widehat{\boldsymbol{x}}^{n})$ a Wardrop equilibrium for the nonatomic game $\Gamma^\infty$ with demands $\boldsymbol{d}^{n}$.
A triangle inequality gives $\norm{X_{e}^{n}-\widehat{x}_{e}}_{L^{2}} \leq\norm{X_{e}^{n}-\widehat{x}_{e}^{n}}_{L^{2}} +\abs{\widehat{x}_{e}^{n}-\widehat{x}_{e}}$, so the result follows from the estimate $\norm{X_{e}^{n}-\widehat{x}_{e}^{n}}_{L^{2}}\leq \theta\sqrt{w^{n}}$
in Theorem~\ref{th:random-load-load-0}, and 
 the bound $\abs{\widehat{\boldsymbol{x}}_{e}^{n}-\widehat{\boldsymbol{x}}_{e}}\leq \Xi \sqrt{\norm{\boldsymbol{d}^{n}-\boldsymbol{d}}_{1}}$ from
 Proposition~\ref{pr:load-loadeq}\ref{it:pr:load-loadeq-2}.
\Halmos
\endproof

\proof{Proof of Theorem~\ref{th:rate-Poisson-0}.}
The bound $\norm{\boldsymbol{x}-\widehat{\boldsymbol{x}}}_{2}\leq \breve{\theta}\sqrt{\Lambda(u)}$
 follows from Proposition~\ref{pr:load-loadeq}\ref{it:pr:load-loadeq-1} because $\boldsymbol{x}$ is an $\varepsilon$-Wardrop equilibrium for 
 $\breve{\Gamma}^{\infty}$ with $\varepsilon=2\kappa\Lambda(u)$. 
Indeed, if $y_{t,s}>0$ there exists some player $i\in\mathscr{N}$ with $t_{i}=t$ and $\widehat{\sigma}_{i}(s)>0$, so that 
the equilibrium condition yields for all $s'\in \mathscr{S}_{t}$
\begin{equation*}
\sum_{e\in s}\Expect_{\widehat{\boldsymbol{\sigma}}}\bracks*{c_{e}(1+Z_{i,e})}
\leq \sum_{e\in s'}\Expect_{\widehat{\boldsymbol{\sigma}}}\bracks*{c_{e}(1+Z_{i,e})},
\end{equation*}
and then \eqref{eq:cP-c1} implies the conditions for $\varepsilon$-Wardrop equilibrium
\begin{equation*}
\sum_{e\in s}\breve{c}_{e}(x_{e})\leq \sum_{e\in s'}\breve{c}_{e}(x_{e})+2\kappa\Lambda(u).
\end{equation*}

Now, taking $V_{e}\sim\Poisson(x_{e})$, Theorem~\ref{th:Barbour}\ref{it:th:Barbour-1} gives
$\rho_{\TV}\parens*{\mathscr{L}(X_{e}),\mathscr{L}(V_{e})}\leq u$,
and \eqref{eq:TV-bound} implies
$\rho_{\TV}({\mathscr{L}(V_{e}),\mathscr{L}(\widehat{X}_{e})})\leq \abs{x_{e}-\widehat{x}_{e}}$,
so that \eqref{eq:B1} follows from a triangle inequality.
\Halmos
\endproof

\proof{Proof of Corollary~\ref{co:rate-Poisson}.}
This follows by considering $\widehat{X}_{e}^{n}\sim\Poisson(\widehat{x}_{e}^{n})$ with $(\widehat{\boldsymbol{y}}^{n},\widehat{\boldsymbol{x}}^{n})$ a Wardrop equilibrium for the game $\breve{\Gamma}^{\infty}$ with demands $\boldsymbol{d}^{n}$, and then using a triangle inequality and applying Theorem~\ref{th:rate-Poisson-0} and Proposition~\ref{pr:load-loadeq}\ref{it:pr:load-loadeq-2}.
\Halmos
\endproof

\section{Convergence of Social Costs}
\label{se:PoA-Proofs}

This appendix includes the auxiliary results on the convergence of social costs required to establish  the convergence of the PoA and the PoS presented in Section~\ref{se:PoA}.

\subsection{Weighted Congestion Games}
\label{suse:PoA-weighted-proofs}

From Section~\ref{se:convergence-weighted} we know that the equilibria in a sequence of weighted congestion games converge to the set of Wardrop equilibria
for the limit game. Below we prove that the corresponding social costs at equilibrium, as well as the optimal social costs, also converge. 
To this end, we start by proving that for any sequence of converging expected flow-load pairs, the social cost of the sequence converges to 
that of the limiting flow-load pair.

\begin{lemma}
\label{le:convergence-SC}
Let $\Gamma_{{\scriptscriptstyle \Wg}}^{n}$ be a sequence of weighted congestion games satisfying the conditions in \eqref{eq:H1}, and $\boldsymbol{\sigma}^{n}$ an arbitrary sequence of mixed strategy profiles (not necessarily equilibria). 
Let $\parens{Y^{n}_{t,s}, X^{n}_{e}}$ be the corresponding random flow-load pairs with expected values $\parens{y^{n}_{t,s},x^{n}_{e}}$.
Then, along any subsequence of $\parens{\boldsymbol{y}^{n},\boldsymbol{x}^{n}}$ converging to some $\parens{\boldsymbol{y},\boldsymbol{x}}$, the expected
social cost $\ESC(\boldsymbol{\sigma}^{n})$ converges to $\SC(\boldsymbol{y},\boldsymbol{x})=\sum_{e\in\mathscr{E}}x_{e}c_{e}(x_{e})$.
\end{lemma}

\proof{Proof.}
Take a convergent subsequence and rename it so that $\parens{\boldsymbol{y}^{n},\boldsymbol{x}^{n}}\to\parens{\boldsymbol{y},\boldsymbol{x}}$. 
By conditioning on $\mathds{1}_{\braces{e\in S_{i}^{n}}}$, which indicates whether player $i$ selects a strategy including $e$, we get
\begin{equation*}
\ESC(\boldsymbol{\sigma}^{n})
=\sum_{e\in\mathscr{E}}\sum_{i\in \mathscr{N}^{n}}\Expect_{\boldsymbol{\sigma}^{n}}\bracks*{w_{i}^{n}\mathds{1}_{\braces{e\in S_{i}^{n}}} \,c_{e}(X_{e}^{n})}
=\sum_{e\in\mathscr{E}}\sum_{i\in \mathscr{N}^{n}}w_{i}^{n}\sigma_{i,e}^{n}\Expect_{\boldsymbol{\sigma}^{n}}\bracks*{c_{e}(X_{i,e}^{n})}.
\end{equation*}

Note that $\abs{X_{i,e}^{n}-X_{e}^{n}}\leq w^{n}\to 0$ and $0\leq X_{e}^{n}\leq d_{\tot}^{n}$ with $d_{\tot}^{n}\to d_{\tot}$.
Because $c_{e}(\argdot)$ is continuous,
and hence uniformly continuous on compact intervals, it follows that $\Expect_{\boldsymbol{\sigma}^{n}}\bracks{c_{e}(X_{i,e}^{n})}-\Expect_{\boldsymbol{\sigma}^{n}}\bracks{c_{e}(X_{e}^{n})}$ converges to $0$ uniformly in $i$, that is,
\begin{equation*}
\delta_{e}^{n}\coloneqq\max_{i\in\mathscr{N}^{n}}\abs*{\Expect_{\boldsymbol{\sigma}^{n}}\bracks*{c_{e}(X_{i,e}^{n})}-\Expect_{\boldsymbol{\sigma}^{n}}\bracks*{c_{e}(X_{e}^{n})}}\to 0.
\end{equation*}
Hence, using the identity 
\begin{equation*}
x_{e}^{n}=\sum_{i\in\mathscr{N}^{n}}w_{i}^{n}\sigma_{i,e}^{n},
\end{equation*}
we obtain
\begin{align*}
\abs*{\ESC(\boldsymbol{\sigma}^{n})-\sum_{e\in\mathscr{E}}x_{e}^{n}\,c_{e}(x_{e}^{n})} 
&\leq \sum_{e\in\mathscr{E}}\sum_{i\in\mathscr{N}^{n}}w_{i}^{n}\sigma_{i,e}^{n}\abs*{\Expect_{\boldsymbol{\sigma}^{n}}\bracks{c_{e}(X_{i,e}^{n})}-c_{e}(x_{e}^{n})}\\
&\leq \sum_{e\in\mathscr{E}}\sum_{i\in \mathscr{N}^{n}}w_{i}^{n}\sigma_{i,e}^{n}\parens*{\delta_{e}^{n}+\abs*{\Expect_{\boldsymbol{\sigma}^{n}}\bracks*{c_{e}(X_{e}^{n})}-c_{e}(x_{e}^{n})}}\\
&= \sum_{e\in\mathscr{E}}x_{e}^{n}\parens*{\delta_{e}^{n}+\abs*{\Expect_{\boldsymbol{\sigma}^{n}}\bracks*{c_{e}(X_{e}^{n})}-c_{e}(x_{e}^{n})}}.
\end{align*}
The conclusion follows because $x_{e}^{n}\to x_{e}$ and $X_{e}^{n}\xrightarrow{\mathscr{D}}x_{e}$, so that $\Expect_{\boldsymbol{\sigma}^{n}}\bracks*{c_{e}(X_{e}^{n})}\to c_{e}(x_{e})$.
\Halmos
\endproof

Using the previous lemma we may derive the convergence of the optimal social cost.

\begin{proposition}
\label{pr:min-SC} 
Let $\Gamma_{{\scriptscriptstyle \Wg}}^{n}$ be a sequence of weighted congestion games satisfying \eqref{eq:H1}. 
Then 
$\Opt(\Gamma_{{\scriptscriptstyle \Wg}}^{n})\to\Opt(\Gamma^{\infty})$.
\end{proposition}

\proof{Proof.}
Let $\parens{\widetilde{\boldsymbol{y}},\widetilde{\boldsymbol{x}}}$ be a social optimum flow-load pair in the limiting game $\Gamma^{\infty}$.
We convert the strategy flow $\widetilde{\boldsymbol{y}}$ into mixed strategies $\fakeopt{\sigma}_{t}\in\bigtriangleup(\mathscr{S}_{t})$ 
by setting 
\begin{equation*}
\forall s\in\mathscr{S}_{t}\qquad \fakeopt{\sigma}_{t}(s)=\widetilde{y}_{t,s}/d_{t} 
\end{equation*}
when $d_{t}>0$, and otherwise taking an arbitrary $\fakeopt{\sigma}_{t}\in\bigtriangleup(\mathscr{S}_{t})$ 
for each type with $d_{t}=0$.

Let $\fakeopt{\boldsymbol{\sigma}}^{n}$ be the strategy profile for $\Gamma_{{\scriptscriptstyle \Wg}}^{n}$ in which player $i$ plays $\fakeopt\sigma_{i}^{n}=\fakeopt\sigma_{t_{i}^{n}}$.
For each $t$ such that $d_{t}=0$ we have $\fakeopt y^{n}_{t,s}\to 0=\widetilde{y}_{t,s}$ for all $s\in\mathscr{S}_{t}$, whereas when $d_{t}>0$ we have
\begin{equation*}
\fakeopt y^{n}_{t,s}=\frac{\widetilde{y}_{t,s}d_{t}^{n}}{d_{t}}\to\widetilde{y}_{t,s}.
\end{equation*}
Hence, $\parens*{\fakeopt y^{n},\fakeopt x^{n}}$ converges to $\parens{\widetilde{y},\widetilde{x}}$, and Lemma~\ref{le:convergence-SC} implies that $\ESC(\fakeopt{\boldsymbol{\sigma}}^{n})\to \SC(\widetilde{\boldsymbol{y}},\widetilde{\boldsymbol{x}})$.

Now, take a sequence $\widetilde{\boldsymbol{\sigma}}^{n}$ of optimal mixed strategies in $\Gamma_{{\scriptscriptstyle \Wg}}^{n}$ and let $(\widetilde{\boldsymbol{y}}^{n},\widetilde{\boldsymbol{x}}^{n})$ be the corresponding expected loads. 
From the optimality of $\widetilde{\boldsymbol{\sigma}}^{n}$ we have $\ESC(\widetilde{\boldsymbol{\sigma}}^{n})\leq\ESC(\fakeopt{\boldsymbol{\sigma}}^{n})$, so that 
\begin{equation}
\label{eq:limsup0}
\limsup_{n\to\infty}\ESC(\widetilde{\boldsymbol{\sigma}}^{n})
\leq \limsup_{n\to\infty}\ESC(\fakeopt{\boldsymbol{\sigma}}^{n})=\SC(\widetilde{\boldsymbol{y}},\widetilde{\boldsymbol{x}})=\Opt(\Gamma^{\infty}).
\end{equation}
On the other hand, taking a subsequence along which we attain the $\liminf_{n\to\infty}\ESC(\widetilde{\boldsymbol{\sigma}}^{n})$ and extracting a 
further subsequence so that $\parens{\widetilde{\boldsymbol{y}}^{n},\widetilde{\boldsymbol{x}}^{n}}$ converges to a certain limit $\parens{\boldsymbol{y},\boldsymbol{x}}$, it follows that
\begin{equation*}
\liminf_{n\to\infty}\ESC(\widetilde{\boldsymbol{\sigma}}^{n})=\SC(\boldsymbol{y},\boldsymbol{x})\geq\Opt(\Gamma^{\infty}),
\end{equation*}
which, combined with \eqref{eq:limsup0}, yields the result.
\Halmos
\endproof

\subsection{Bernoulli Congestion Games}
\label{suse:PoA-BCG-proofs}

The following are the analogous results for a sequence of Bernoulli congestion games. 

\begin{lemma}
\label{le:convergence-SC-Poisson}
Let $\Gamma_{{\scriptscriptstyle \Bg}}^{n}$ be a sequence of Bernoulli congestion games satisfying the conditions in \eqref{eq:H3} and \eqref{eq:HP}, and let $\boldsymbol{\sigma}^{n}$ be an arbitrary sequence of mixed strategies. 
Let $Y^{n}_{t,s}$ and $ X^{n}_{e}$ be the
corresponding random loads with expected values $y^{n}_{t,s}$ and $ x^{n}_{e}$.
Then, along any subsequence of $(\boldsymbol{y}^{n},\boldsymbol{x}^{n})$ converging to some $(\boldsymbol{y}, \boldsymbol{x})$, the expected social cost $\ESC(\boldsymbol{\sigma}^{n})$ converges 
to $\breve{\SC}(\boldsymbol{y},\boldsymbol{x})\coloneqq
\sum_{e\in\mathscr{E}}x_{e}\, \breve c_{e}(x_{e})$.
\end{lemma}

\proof{Proof.} 
Take a convergent subsequence and rename it so that $(\boldsymbol{y}^{n}, \boldsymbol{x}^{n})\to (\boldsymbol{y}, \boldsymbol{x})$.
By conditioning on the event $U_{i,e}^{n}=1$, we have
\begin{equation*}
\ESC(\boldsymbol{\sigma}^{n})
=\sum_{e\in\mathscr{E}}\sum_{i\in \mathscr{N}^{n}}\Expect_{\boldsymbol{\sigma}^{n}}\bracks*{U_{i,e}^{n} \,c_{e}(X_{e}^{n})}
=\sum_{e\in\mathscr{E}}\sum_{i\in \mathscr{N}^{n}}u_{i}^{n}\sigma_{i,e}^{n}\,\Expect_{\boldsymbol{\sigma}^{n}}\bracks*{c_{e}(1+Z_{i,e}^{n})}.
\end{equation*}
Using the identity $x_{e}^{n}=\sum_{i\in\mathscr{N}^{n}}u_{i}^{n} \sigma_{i,e}^{n}$ and invoking Lemma~\ref{le:aux}, we obtain
\begin{align*}
\abs*{\ESC(\boldsymbol{\sigma}^{n})-\sum_{e\in\mathscr{E}}x_{e}^{n}\, \breve{c}_{e}(x_{e}^{n})}&
\leq \sum_{e\in\mathscr{E}}\sum_{i\in \mathscr{N}^{n}}u_{i}^{n}\sigma_{i,e}^{n}\abs*{\Expect_{\boldsymbol{\sigma}^{n}}\bracks{c_{e}(1+Z_{i,e}^{n})}-\breve c_{e}(x_{e}^{n})}\\
&\leq \sum_{e\in\mathscr{E}}\sum_{i\in \mathscr{N}^{n}}u_{i}^{n}\sigma_{i,e}^{n}\,\Lambda(u^{n})\\ 
&= \sum_{e\in\mathscr{E}}x_{e}^{n}\,\Lambda(u^{n})\to 0,
\end{align*}
and then the conclusion follows from $x_{e}^{n} \to x_{e}$ and the continuity of $\breve{c}_{e}(\argdot)$.
\Halmos
\endproof

\begin{proposition}
\label{pr:min-SC-Poisson} 
Let $\Gamma_{{\scriptscriptstyle \Bg}}^{n}$ be a sequence of Bernoulli congestion games satisfying \eqref{eq:H3} and\eqref{eq:HP}. 
Then 
$\Opt(\Gamma_{{\scriptscriptstyle \Bg}}^{n})\to\Opt(\breve{\Gamma}^{\infty})$. 
\end{proposition}

\proof{Proof.}
It suffices to repeat the proof of Proposition~\ref{pr:min-SC} step-by-step,
replacing $\SC(\argdot)$ with $\breve{\SC}(\argdot)$ and invoking Lemma~\ref{le:convergence-SC-Poisson} instead of Lemma~\ref{le:convergence-SC}.
\Halmos
\endproof

\section{Poisson Approximation for Sums of Bernoulli Random Variables}
\label{se:Poisson-Bernoulli}

This section collects some known facts on the Poisson approximation for sums of Bernoulli random variables. 
The main results are taken from \citet{AdeLeu:B2005}, \citet{BarHal:MPCPS1984}, and \citet{BorRuz:AP2002}, suitably adapted to our goals.
We recall that $X\sim\Poisson(x)$ with parameter $x\geq 0$ if and only if 
\begin{equation*}
\Prob(X=k)=\mathrm{e}^{-x}\frac{x^{k}}{k!}\quad\forall k\in\mathbb{N}.
\end{equation*}
As usual we denote the law of $X$ by $\mathscr{L}(X)$.

Two Poisson variables $X\sim\Poisson(x)$ and $Y\sim\Poisson(y)$ are close when $x\approx y$.
In fact, their total variation distance (see  \eqref{eq:TV1}) can be estimated as \citep[see][]{AdeLeu:B2005}
\begin{equation}
\label{eq:TV-bound}
\rho_{\TV}(\mathscr{L}(X),\mathscr{L}(Y))\leq 1-\mathrm{e}^{-\abs{x-y}}\leq \abs{x-y}.
\end{equation}

Given a function $h:\mathbb{N}\to\mathbb{R}$, for each $X\sim\Poisson(x)$ with $\Expect\bracks{\abs{h(X)}}<\infty$, we define
\begin{equation*}
\breve h(x)=\Expect\bracks{h(X)}=\sum_{k=0}^{\infty} h(k) \mathrm{e}^{-x}\frac{x^{k}}{k!}.
\end{equation*}
We also define $h\mapsto \Delta h$ as the operator that takes the function $h$ into $\Delta h(k)=h(k+1)-h(k)$.
Below we state three useful facts that are used in the sequel.
\begin{proposition}
\label{pr:Poisson-Delta} 
Let $X\sim\Poisson(x)$. Then
\begin{enumerate}[\upshape(a)]
\item
\label{it:pr:Poisson-Delta-1} 
$\Expect\bracks{\abs{h(Y)}}\leq \mathrm{e}^{x-y}\,\Expect\bracks{\abs{h(X)}}$ for each $Y\sim\Poisson(y)$ with $y\leq x$.

\item
\label{it:pr:Poisson-Delta-2} 
For each $j=1,2,\ldots$ we have\footnote{Note that if $\Expect\bracks{X^j\abs{h(X)}}<\infty$ holds for a certain $j$, it also holds for $j'=1,\ldots,j$.}
\begin{equation*}
\Expect\bracks{\abs{\Delta^j h(X)}}<\infty \iff \Expect\bracks{\abs{h(X+j)}}<\infty  \iff \Expect\bracks{X^j\abs{h(X)}}<\infty.
\end{equation*}
\end{enumerate}
\end{proposition}

\proof{Proof.} 
Property \ref{it:pr:Poisson-Delta-1} is just the monotonicity of 
$x\mapsto \mathrm{e}^{x}\cdot\Expect\bracks{\abs{h(X)}}=\sum_{k=0}^{\infty}\abs{h(k)}\frac{x^{k}}{k!}$, whereas 
\ref{it:pr:Poisson-Delta-2} can be found in \citet[proposition~1]{BorRuz:AP2002}.
\Halmos
\endproof

\begin{proposition} 
\label{pr:Poisson-Delta-j}
Let $h:\mathbb{N}\to\mathbb{R}$ and $V\sim\Poisson(d_{\max})$.
\begin{enumerate}[\upshape(a)]
\item \label{it:pr:Poisson-Delta-j-1}
If\hspace{0.2ex} $\Expect\bracks{\abs{h(V)}}<\infty$ then $\breve h(x)$ is well defined and continuous for $x\in[0,d_{\max}]$. 

\item \label{it:pr:Poisson-Delta-j-2}
If\hspace{0.2ex} $\Expect\bracks{\abs{\Delta^j h(V)}}<\infty$ for some $j\in\mathbb{N}$ then $\breve h$ is of class $C^j$ on $[0,d_{\max}]$ and $\breve h^{(j)}(x)=\Expect\bracks{\Delta^j h(X)}$.
\end{enumerate}
\end{proposition}

\proof{Proof.}
\ref{it:pr:Poisson-Delta-j-1} 
It suffices to note that the series 
\begin{equation*}
f(x)=\mathrm{e}^{x}\breve h(x)=\sum_{k=0}^{\infty} h(k)\frac{x^{k}}{k!}
\end{equation*}
is well defined and continuous. 
This follows because the partial sums $f_{n}(x)=\sum_{k=0}^{n} h(k)x^{k}/k!$ converge uniformly to $f(x)$. 
Indeed,
\begin{equation*}
\sup_{x\in[0,d_{\max}]}\abs*{f(x)-f_{n}(x)}\leq \sup_{x\in[0,d_{\max}]}\sum_{k=n+1}^{\infty} \abs{h(k)}\frac{x^{k}}{k!}\leq \sum_{k=n+1}^{\infty} \abs{h(k)} \frac{d_{\max}^{k}}{k!},
\end{equation*}
where the latter tail of the series tends to 0 as $n\to\infty$ because $\Expect\bracks{\abs{h(V)}}<\infty$.

\medskip
\noindent
\ref{it:pr:Poisson-Delta-j-2} 
Consider first the case $j=1$. 
We note that the derivatives 
\begin{equation*}
f_{n}'(x)=\sum_{k=0}^{n-1} h(k+1)\frac{x^{k}}{k!}
\end{equation*}
converge uniformly towards $g(x)=\sum_{k=0}^{\infty} h(k+1)x^{k}/k!$.
This follows from part \ref{it:pr:Poisson-Delta-j-1} because by Proposition~\ref{pr:Poisson-Delta}\ref{it:pr:Poisson-Delta-2} with $j=1$ we have
$\Expect\bracks{\abs{h(V+1)}}<\infty$.
Hence $f$ is $C^{1}$ with $f'(x)=g(x)$ and then
$\breve h(x)=\mathrm{e}^{-x}f(x)$ is also in $C^{1}$ with
\begin{equation*}
\breve{h}'(x)=\mathrm{e}^{-x}f'(x)-\mathrm{e}^{-x}f(x)=\sum_{k=0}^{\infty} \big(h(k+1)-h(k)\big)\mathrm{e}^{-x}\frac{x^{k}}{k!}=\Expect\bracks{\Delta h(X)}.
\end{equation*}
This establishes the case $j=1$. Applying this property to $\Delta h$ we obtain the result for $j=2$, and then the cases $j=3,4, \ldots$ follow by induction.
\Halmos
\endproof

\begin{corollary} 
\label{co:Poisson-alpha} 
Let $V\sim\Poisson(d_{\max})$ and suppose that $\Expect\bracks{\abs{\Delta^{2} h(V)}}\leq\nu<\infty$. 
\begin{enumerate}[\upshape(a)]
\item
\label{co:Poisson-alpha-a}
 For all $x\in[0,d_{\max}]$ we have $\abs{\breve{h}'(x)} \leq (\mathrm{e}^{d_{\max}}-1)\,\nu+\abs{h(1)-h(0)}$.
\item
\label{co:Poisson-alpha-b}
 If $h(\argdot)$ is weakly increasing then for all $x\in[0,d_{\max}]$ we have $\breve{h}'(x)\geq 0$, with strict inequality when $h(\argdot)$ is nonconstant.
\end{enumerate}
\end{corollary}

\proof{Proof.} 
\ref{co:Poisson-alpha-a} Combining  Propositions~\ref{pr:Poisson-Delta-j}\ref{it:pr:Poisson-Delta-j-2} and \ref{pr:Poisson-Delta}\ref{it:pr:Poisson-Delta-1} we get
$|\breve{h}''(x)|\leq \mathrm{e}^{d_{\max}-x}\nu$, so that by integration it follows that $|\breve{h'}(x)-\breve{h}'(0)|\leq (\mathrm{e}^{d_{\max}}-1)\nu$,
and we conclude because $\breve{h}'(0)=h(1)-h(0)$.

\vspace{1ex}
\noindent
\ref{co:Poisson-alpha-b} This follows from Proposition~\ref{pr:Poisson-Delta-j}\ref{it:pr:Poisson-Delta-j-2}.
\Halmos
\endproof

Let $S=X_{1}+\ldots+X_{n}$ be a sum of independent Bernoulli random variables with $\Prob(X_{i}=1)=p_{i}\in (0,1)$, and denote 
$x=\Expect[S]=p_{1}+\ldots+p_{n}$. Consider a Poisson variable 
$X\sim\Poisson(x)$ with the same expectation.
The following result shows that $S$ and $X$ are close when the $p_{i}$'s are small.

\begin{theorem}
\label{th:Barbour}
Let $p=\max\braces{p_{1},\dots,p_{n}}$. Then
\begin{enumerate}[\upshape(a)]
\item
\label{it:th:Barbour-1}
The following double inequality holds:
\begin{equation*}
\rho_{\TV}(\mathscr{L}(S),\mathscr{L}(X))\leq (1-\mathrm{e}^{-x})\,x^{-1}\sum_{i=1}^{n}p_{i}^{2}\leq p.
\end{equation*}

\item
\label{it:th:Barbour-2}
If \ $h:\mathbb{N}\to\mathbb{R}$ is such that $\Expect\bracks{\abs{\Delta^{2} h(X)}}\leq\nu <\infty$, then 
\begin{equation*}
\abs*{\Expect\bracks*{h(S)}-\Expect\bracks*{h(X)}}\leq\frac{\nu\, x}{2}\frac{p\, \mathrm{e}^{p}}{(1-p)^{2}}.
\end{equation*}
\end{enumerate}
\end{theorem}

\proof{Proof.}
These properties follow from \citet[theorem~1]{BarHal:MPCPS1984} and \citet[corollary~4]{BorRuz:AP2002}, respectively.
\Halmos
\endproof

\section{List of Symbols}
\label{se:symbols}

The following table contains the symbols used in the paper.

\begin{longtable}{p{.10\textwidth} p{.85\textwidth}}

$\BNE(\Gamma_{{\scriptscriptstyle \Bg}}^{n})$ & set of Bayesian Nash equilibria of the game $\Gamma_{{\scriptscriptstyle \Bg}}^{n}$\\
$c_{e}$ & cost function of edge $e$\\
$\breve{c}_{e}$ & cost function of edge $e$ in the Poisson limit game, defined in \eqref{eq:cost-Poisson}\\
$C$ & upper bound for $\sum_{e\in s}c_{e}(d_{\max})$ in Proposition~\ref{pr:load-loadeq} and $\sum_{e\in s}\breve{c}_{e}(d_{\max})$ in Corollary~\ref{co:rate-Poisson}\\
$C_{s}$ & $\sum_{e\in s} c_{e}(x_{e})$\\
$c_{\min}'$ & lower bound on $c_{e}'(\argdot)$\\
$c_{\max}'$ & upper bound on $c_{e}'(\argdot)$\\
$\breve{c}_{\min}'$ & lower bound on $\breve{c}_{e}'(\argdot)$\\
$\boldsymbol{d}$ & demand vector\\
$d_{t}$ & demand of type $t$\\
${d}^{n}_{t}$ & aggregate demand of type $t$ in the games $\Gamma_{{\scriptscriptstyle \Wg}}^{n}$ and $\Gamma_{{\scriptscriptstyle \Bg}}^{n}$\\
$d_{\max}$ & upper bound on the demand\\
$d_{\tot}$ & total demand\\
$D$ & random demand\\
$e$ & edge\\
$\mathscr{E}$ & set of resources\\
$\Expect_{\boldsymbol{\sigma}}$ & expectation induced by $\boldsymbol{\sigma}$\\
$\Eq(\Gamma^{\infty})$ & equilibrium cost of $\Gamma^{\infty}$\\
$\ESC(\boldsymbol{\sigma})$ & expected social cost of $\boldsymbol{\sigma}$, defined in \eqref{eq:ESC}\\
$\mathscr{F}(\boldsymbol{d})$ & feasible pairs for demand $\boldsymbol{d}$\\
$\mathscr{G}$  & $\parens{\mathscr{E},(c_{e})_{e\in\mathscr{E}},\mathscr{T},(\mathscr{S}_{t})_{t\in\mathscr{T}}}$, defined in \eqref{eq:structural}\\
$L$ & Lipschitz constant for $c_{e}'(\argdot)$\\ 
$\mathscr{L}(X)$ & law of the random variable $X$\\
$\MNE(\Gamma)$ & set of mixed Nash equilibria of $\Gamma$\\
$\boldsymbol{n}$ & $\parens{n_{t,s}}_{t\in\mathscr{T},s\in\mathscr{S}_{t}}$\\
$N_{e}$ & $\sum_{i\colon e\in s_{i}}U_{i}$, i.e., random number of players who use resource $e$\\
$N_{t}$ & random number of players of type $t$ in a Poisson game\\
$\boldsymbol{N}$ & $(N_{t})_{t\in\mathscr{T}}$\\
$\mathscr{N}$ & set of players\\
$\Opt(\Gamma^{\infty})$ & optimum social cost in $\Gamma^{\infty}$\\
$\Opt(\Gamma)$ & optimum expected social cost in $\Gamma$\\
$\Prob_{\boldsymbol{\sigma}}$ & probability measure induced by $\boldsymbol{\sigma}$\\
$\PNE(\Gamma)$ & set of pure Nash equilibria of $\Gamma$\\
$\PoA$ & price of anarchy\\
$\PoS$ & price of stability\\
$\mathscr{S}_{t}$ & set of strategies for type $t$\\
$\boldsymbol{s}$ & strategy profile\\ 
$s_{i}$ & strategy of player $i$\\
$S_{i}$ & random strategy of player $i$\\
$\SC(\boldsymbol{y},\boldsymbol{x})$ & social cost of $(\boldsymbol{y},\boldsymbol{x})$, defined in \eqref{eq:SC}\\
$\mathscr{T}$ & set of types\\
$t_{i}$ & type of player $i$\\
$u_{i}$ & probability of player $i$ being active\\
$u^{n}$ & $\max_{i\in\mathscr{N}^{n}}u^{n}_{i}$, defined in \eqref{eq:prob-part-to-0}\\
$U_{i}$ & indicator of player $i$ being active\\
$U_{i,e}$ & $U_{i}\mathds{1}_{\braces{e\in S_{i}}}$\\
$w_{i}$ & weight of player $i$\\
$w^{n}$ & $\max_{i\in\mathscr{N}^{n}}w_{i}^{n}$, defined in \eqref{eq:w-n}\\
$\boldsymbol{x}$ & load vector \\
$x_{e}$ & load of edge $e$\\
$\boldsymbol{X}$ & random load vector\\
$X_{e}$ & random load of edge $e$\\
$X_{i,e}$ & $w_{i}+\sum_{j\neq i}w_{j}\mathds{1}_{\braces{e\in S_{j}}}$, defined in \eqref{eq:load-e-i-present}\\
$\widetilde{\boldsymbol{x}}$ & optimum load in $\Gamma^{\infty}$\\
$(\widehat{\boldsymbol{y}},\widehat{\boldsymbol{x}})$ & equilibrium flow-load pair\\
$\boldsymbol{y}$ & flow vector\\
$y_{t,s}$ & flow of type $t$ on strategy $s$\\
$\boldsymbol{Y}$ & random  flow vector\\
$Y_{t,s}$ & random flow of type $t$ on strategy $s$\\
$Z_{i,e}$ & $\sum_{j\neq i}U_{j,e}$, defined in \eqref{eq:Z}\\
$\Gamma$ & game\\
$\Gamma_{{\scriptscriptstyle \Bg}}$ & Bernoulli congestion game\\
$\Gamma_{{\scriptscriptstyle \Pg}}$ & game with population uncertainty\\
$\Gamma_{{\scriptscriptstyle \Wg}}$ & weighted congestion game\\
$\Gamma^{\infty}$ & nonatomic congestion game\\
$\Delta$ & difference operator\\
$\zeta$ & $(\mathrm{e}^{d_{\max}}-1)\nu+\max_{e\in\mathscr{E}}\,(c_{e}(2)- c_{e}(1))$, defined in \eqref{eq:zeta}\\
$\eta_{t}$ & $\min_{s\in\mathscr{S}_{t}} C_{s}$\\
$\theta$ & $\sqrt{d_{\max}/4}+\sqrt{2\kappa d_{\max} \parens*{c_{\max}'+L d_{\max}/4}/c_{\min}'}$, defined in Theorem~\ref{th:random-load-load-0}\\
$\breve{\theta}$ & $\sqrt{2\kappa d_{\max}/\breve{c}_{\min}'}$, defined in Theorem~\ref{th:rate-Poisson-0}\\
$\kappa$ & cardinality of the largest feasible strategy $s\in\cup_{t\in\mathscr{T}}\mathscr{S}_{t}$, defined in Theorem~\ref{th:random-load-load-0}\\
$\Lambda(u)$ & $\frac{d_{\max}\nu}{2}\frac{u \mathrm{e}^{u}}{(1-u)^{2}}+\zeta u$, defined in \eqref{eq:Lambda}\\
$\mu$ & joint probability distribution of the number of players in a Poisson game\\
$\nu$ & upper bound of $\Expect\bracks*{\abs{\Delta^{2} c_{e}(1+V)}}$ in \eqref{eq:HP}\\
$\Xi$ & $\sqrt{2 C/c_{\min}'}$, defined in Proposition~\ref{pr:load-loadeq}\\
$\breve{\Xi}$ & $\sqrt{2 C/\breve{c}_{\min}'}$, defined in  Corollary\ref{co:rate-Poisson}\\
$\rho_{\TV}$ & total variation distance, defined in \eqref{eq:TV1}\\
$\boldsymbol{\sigma}$ & mixed strategy profile\\
$\sigma_{i}$ & mixed strategy of player $i$\\
$\sigma_{i,e}$ & $\Prob_{\boldsymbol{\sigma}}(e\in S_{i})$, defined in \eqref{eq:prob-e-in-Si}\\
$\widehat{\boldsymbol{\sigma}}$ & mixed Nash equilibrium\\
$\widetilde{\boldsymbol{\sigma}}$ & optimum mixed strategy in $\Gamma$\\
$\Sigma$ & set of mixed strategy profiles\\

$\bigtriangleup(\mathscr{S}_{t})$ & simplex of probability measures on $\mathscr{S}_{t}$\\
$\norm{\boldsymbol{x}}_{2}$ & $\sqrt{\sum_{e\in\mathscr{E}}x_{e}^{2}}$\\
$\norm{\boldsymbol{d}}_{1}$ & $\sum_{t\in\mathscr{T}}\abs{d_{t}}$\\ 
$\norm{X}_{L^{2}}$ & $\sqrt{\Expect\bracks{X^{2}}}$\\

\end{longtable}

\end{APPENDICES}

\section*{Acknowledgments.}
This collaboration started in Dagstuhl at the Seminar on Dynamic Traffic Models in Transportation Science in 2018.
Marco Scarsini is a member of the Gruppo Nazionale per l’Analisi Matematica, la Probabilit\`a e le loro Applicazioni, Istituto Nazionale di Alta Matematica.

\paragraph{Funding:}
Roberto Cominetti gratefully acknowledges the support of Proyecto Anillo ANID/PIA/ACT192094.
Marco Scarsini's work was partially supported by the Gruppo Nazionale per l'Analisi Matematica, la Probabilit\`a e le loro Applicazioni, Istituto Nazionale di Alta Matematica [Grant 2020 ``Random walks on random games''] and the Progetti di Rilevante Interesse Nazionale 2017 [Grant ``Algorithms, Games, and Digital Markets''].
This research also received partial support from the European Cooperation in Science and Technology [Action European Network for Game Theory].


\bibliographystyle{informs2014} 
\bibliography{bibcong-MOR} 


\end{document}